\documentclass[draftclsnofoot,onecolumn]{IEEEtran}

\usepackage{amsmath,amssymb,amsthm}
\usepackage{stfloats}

\newtheorem{lemma}{Lemma}%[section]
\newtheorem{example}{Example}
\newtheorem{definition}{Definition}

\begin{document}

\title{Constructing MSR codes with subpacketization $2^{n/3}$ for $k+1$ helper nodes}

\author{Ningning Wang, Guodong Li, Sihuang Hu,~\IEEEmembership{Member,~IEEE,} and Min Ye
\thanks{Research partially funded by National Key R\&D Program of China under Grant No. 2021YFA1001000, National Natural Science Foundation of China under Grant No. 12001322 and 12231014, Shandong Provincial Natural Science Foundation under Grant No. ZR202010220025, Shenzhen Stable Support program under Grant No. WDZC20220811170401001, RISC-V International Open Source Laboratory, and a Taishan scholar program of Shandong Province. This paper was presented in part at the 2022 IEEE International Symposium on Information Theory~\cite{LWHY22}.}
\thanks{Ningning Wang, Guodong Li and Sihuang Hu are with  Key Laboratory of Cryptologic Technology and Information Security, Ministry of Education, Shandong University, Qingdao, Shandong, 266237, China and School of Cyber Science and Technology, Shandong University, Qingdao, Shandong, 266237, China. S. Hu is also with Quancheng Laboratory, Jinan 250103, China. Email: \{nningwang,guodongli\}@mail.sdu.edu.cn, husihuang@sdu.edu.cn}
\thanks{Min Ye is with Tsinghua-Berkeley Shenzhen Institute, Tsinghua Shenzhen International Graduate School, Shenzhen 518055, China. Email: yeemmi@gmail.com}}

%\markboth{Journal of \LaTeX\ Class Files,~Vol.~14, No.~8, August~2015}%
%{Shell \MakeLowercase{\textit{et al.}}: Constructing MSR codes with subpacketization $2^{n/3}$ for $k+1$ helper nodes}

\maketitle

\begin{abstract}
Wang et al. (IEEE Transactions on Information Theory, vol. 62, no. 8, 2016) proposed an explicit construction of an $(n=k+2,k)$ Minimum Storage Regenerating (MSR) code with $2$ parity nodes and subpacketization $2^{k/3}$. The number of helper nodes for this code is $d=k+1=n-1$, and this code has the smallest subpacketization among all the existing explicit constructions of MSR codes with the same $n,k$ and $d$. In this paper, we present a new construction of MSR codes for a wider range of parameters. More precisely, we still fix $d=k+1$, but we allow the code length $n$ to be any integer satisfying $n\ge k+2$. The field size of our code is linear in $n$, and the subpacketization of our code is $2^{n/3}$. This value is slightly larger than the subpacketization of the construction by Wang et al. because their code construction only guarantees optimal repair for all the systematic nodes while our code construction guarantees optimal repair for all nodes.
\end{abstract}


\section{Introduction}\label{sect:intro}
\IEEEPARstart{D}{imakis} et al. initiated the study of repairing MDS codes in distributed storage systems \cite{Dimakis10}, where the main objective is to construct MDS codes that can efficiently repair any single node failure. In a distributed storage system, an $(n,k)$ MDS array code consists of $k$ information nodes and $r=n-k$ parity nodes. Each node of the array code is a length-$\ell$ vector over some finite field $\mathbb{F}$, where the parameter $\ell$ is called the subpacketization or the node size of the array code. When there is a single node failure, one may connect to $d$ helper nodes and download part of the information stored on these nodes to recover the failed node, where the value of $d$ satisfies $k\le d\le n-1$. The amount of downloaded data in the repair procedure is called the repair bandwidth. For an $(n,k)$ MDS array code with subpacketization $\ell$, it was shown in \cite{Dimakis10} that the repair bandwidth is at least $\frac{d\ell}{d-k+1}$ when the repair procedure involves $d$ helper nodes. This lower bound is called the cut-set bound. An MDS code that achieves the optimal repair bandwidth for the repair of any single failed node from any $d$ helper nodes is called an Minimum Storage Regenerating (MSR) code.

The first explicit construction of MSR codes is the product matrix construction \cite{Rashmi11}. Yet it is limited to the low code rate regime. After that, explicit constructions of high-rate MSR codes were given in \cite{Tamo13,Wang16,Ye16,Ye16a,Sasid16,Raviv17,Tamo17RS,Li18,Duursma21,Vajha21}. The complete list of constructions and parameters of MSR codes can be find in recent paper~\cite{Kumar18} and \cite{LLT2022}. Among these MSR code constructions, the most relevant one to this paper is the $(n=k+2,k)$ MSR code construction with subpacketization $2^{k/3}$ proposed in \cite{Wang16}. More precisely, Wang et al. \cite{Wang16} proposed a framework of constructing MDS array codes with subpacketization $r^{k/(r+1)}$ that can optimally repair all the systematic nodes from $d=n-1$ helper nodes. For the special case of $r=2$, they further gave explicit constructions of such codes over any finite field whose size is larger than $\frac{2}{3}k$. This construction has the smallest subpacketization among all the existing MSR code constructions with parameters $d=k+1$ and $n=k+2$.

In this paper, we propose a new construction of MSR codes for which we relax the condition $d=n-1$ in the $(n=k+2,k)$ MSR code constructed in \cite{Wang16}. Specifically, the number of helper nodes in \cite{Wang16} is $d=k+1=n-1$. In our new construction, we still fix $d=k+1$, but we allow $n$ to be any integer satisfying $n\ge k+2$. Our code can be constructed over finite fields whose size is linear in $n$, and the subpacketization of our code is $2^{n/3}$. This value is slightly larger than the subpacketization in \cite{Wang16} because their code construction only guarantees optimal repair for all the systematic nodes while our code construction guarantees optimal repair for all nodes.

In the discussion of our new code construction above, we have implicitly assumed that $n$ is divisible by $3$. If $n$ is not a multiple of $3$, then the subpacketization of our new code becomes $2^{\lceil n/3 \rceil}$. Since the code construction in the non-divisible case can be easily obtained by puncturing/removing some nodes from the construction in the divisible case, in this paper we will focus on the case where $n$ is a multiple of $3$.

One main technical obstacle of this paper is to prove the invertibility of the matrix $Q$ in Lemma~\ref{lm:det}.
We explicitly compute the determinant of $Q$ by first treating it as a multivariable polynomial of the entries (seen as variables) of $Q$,
and then extract as many polynomial factors as we can using elementary column operations.
The choice of parameters in our code construction guarantees that all the polynomial factors are nonzero. 
At last by comparing the terms in the expansion of the determinant and the product of these polynomial factors,
we conclude that the determinant is equal to the product of these polynomial factors up to some coefficient $c\in\{-1,1\}$,
and hence nonzero. 

%Lemma~\ref{lm:det} is the main technical contribution of this paper where we prove the certain matrix can be invertible by computing its determinate. First, we can see that the determinant of this matrix is a multivariable polynomial which can observe the degree of each variable. Then we can get enough factors of its determinant through column transformation to match the maximum degree of each variable. At last, the constant coefficient of the determinant can be obtained by analysis. This completes the factorization of the matrix determinant.

Finally, we note that a recent work \cite{Vajha21} proposed a generalization of the MSR code construction in \cite{Ye16a,Sasid16,Li18}, and their generalization of \cite{Ye16a,Sasid16,Li18} is somewhat similar to the generalization of \cite{Wang16} in this paper. More precisely, both the construction in \cite{Ye16a,Sasid16,Li18} and the construction in \cite{Wang16} require $d=n-1$. Vajha et al. \cite{Vajha21} generalized the construction in \cite{Ye16a,Sasid16,Li18} to the case $d<n-1$ while we generalize the construction in \cite{Wang16} to the case $d<n-1$. However, the construction in \cite{Ye16a,Sasid16,Li18} is very different from the construction in \cite{Wang16}, and the generalization in this paper also uses different techniques from the generalization in \cite{Vajha21}. For $d=k+1$, the subpacketization in \cite{Vajha21} is $2^{n/2}$. This value is larger than the subpacketization $2^{n/3}$ in our paper, but the MSR codes constructed in \cite{Vajha21} guarantee that both the amount of accessed data and the repair bandwidth in the repair procedure achieve the cut-set bound while the MSR codes in this paper only guarantee the optimal repair bandwidth. 

%The structure of this paper is as follows: Section $\ref{construction}$ gives the construction of the code, Section $\ref{example}$ gives an example, Section $\ref{MDS}$ gives the proof of MDS property in the general case of the code, Section $\ref{repair}$ gives the repair scheme of the code in the general case. 

\section{New code construction} \label{construction}
We write a codeword of an $(n,k)$ array code with subpacketization $\ell$ as $(C_0,C_1,\dots,C_{n-1})$, where each $C_i=(C_i(0),C_i(1),\dots,C_i(\ell-1))^T$ is a vector of length $\ell$ over some finite field $\mathbb{F}$. We define an $(n,k)$ array code by the following set of parity check equations: for $a=0,1,\dots,\ell-1$,
\begin{equation} \label{eq:pc}
{\scriptsize\begin{aligned}
\sum_{i=0}^{n-1} \Big( A_i(a,0) C_i(0) + A_i(a,1) C_i(1) + \dots + A_i(a,\ell-1) C_i(\ell-1) \Big) = \mathbf{0} .
\end{aligned}}
\end{equation}
Each $A_i(a,b)$ in the above equation is a column vector of length $r=n-k$ for all $0\le a,b\le \ell-1$. The $\mathbf{0}$ on the right-hand side of the equation denotes the all-zero column vector of length $r$~\footnote{Throughout this paper, we will use $0$ to denote the number $0$ or the all-zero matrix, and use bold $\mathbf{0}$ to denote the all-zero column vector. The size of vector and matrix can be easily derived from the context.}. 
In our construction, we set $\ell=2^{n/3}$, so for every $a\in\{0,1,\dots,\ell-1\}$, we can write its $n/3$-digit binary expansion as $a=(a_{n/3-1},a_{n/3-2},\dots,a_0)$, where $a_0$ is the least significant bit.

We divide the $n$ nodes $(C_0,C_1,\dots,C_{n-1})$ into $n/3$ groups of size $3$. More precisely, for every $0\le i\le n/3-1$, the three nodes $C_{3i},C_{3i+1},C_{3i+2}$ are in the same group. For the index $3i+j$ of each node, the value $i$ is the index of the group to which this node belongs, and $j\in\{0,1,2\}$ is the index of the node within the group. Let $\mathbb{F}$ be a finite field with size $|\mathbb{F}|\ge 2n+1$. 
Let $\lambda_0,\lambda_1,\dots,\lambda_{2n-1}$ be $2n$ distinct elements in $\mathbb{F}$ satisfying 
\begin{equation} \label{eq:crcd}
	\begin{aligned}
		&\frac{(\lambda_{6i+1}-\lambda_{6i+3})(\lambda_{6i+1}-\lambda_{6i+5})}{(\lambda_{6i+1}-\lambda_{6i})(\lambda_{6i+1}-\lambda_{6i+4})} \\
		\neq  &\frac{(\lambda_{6i+2}-\lambda_{6i+3})(\lambda_{6i+2}-\lambda_{6i+5})}{(\lambda_{6i+2}-\lambda_{6i})(\lambda_{6i+2}-\lambda_{6i+4})}  \text{~for all~} 0\le i\le n/3-1.
	\end{aligned}
\end{equation}
For $0\le i\le n/3-1$, we introduce the shorthand notation
\begin{equation}  \label{eq:gama}
	\begin{aligned}
		&\gamma_{6i+1}=-\frac{(\lambda_{6i+1}-\lambda_{6i+3})(\lambda_{6i+1}-\lambda_{6i+5})}{(\lambda_{6i+1}-\lambda_{6i})(\lambda_{6i+1}-\lambda_{6i+4})},  \\
		&\gamma_{6i+2}=-\frac{(\lambda_{6i+2}-\lambda_{6i+3})(\lambda_{6i+2}-\lambda_{6i+5})}{(\lambda_{6i+2}-\lambda_{6i})(\lambda_{6i+2}-\lambda_{6i+4})},
	\end{aligned}
\end{equation}
which will be used in later sections.
The condition \eqref{eq:crcd} is equivalent to
\begin{equation} \label{eq:cond}
\gamma_{6i+1} \neq \gamma_{6i+2} \quad \text{for all~} 0\le i\le n/3-1 .
\end{equation}
In Lemma~\ref{lm:1q} below, we show that as long as $|\mathbb{F}|\ge 2n+1$, one can always choose $2n$ distinct $\lambda_i$'s satisfying \eqref{eq:cond}.
We further write 
\begin{equation} \label{eq:dli}
L_i := \begin{bmatrix}
1 \\
\lambda_i \\
\lambda_i^2 \\
\vdots \\
\lambda_i^{r-1}
\end{bmatrix} .
\end{equation}
Now we are ready to define the vectors $A_i(a,b)$ in \eqref{eq:pc}. For every $0\le i\le n/3-1$ and every $0\le a,b\le \ell-1$, define
\begin{equation} \label{pcmt}
\begin{aligned}
A_{3i}(a,b) &=\left\{
\begin{array}{ll}
  L_{6i+a_i}   & \mbox{if~} a=b \\
  L_{6i}-L_{6i+1}   & \mbox{if~} a_i=0, b_i=1, \\ &\mbox{and~} a_j=b_j \,\forall j\neq i \\
  \mathbf{0} & \mbox{otherwise},
\end{array}
\right.  \\
A_{3i+1}(a,b) &=\left\{
\begin{array}{ll}
  L_{6i+2+a_i}   & \mbox{if~} a=b \\
  L_{6i+3}-L_{6i+2}   & \mbox{if~} a_i=1, b_i=0, \\ &\mbox{and~} a_j=b_j \,\forall j\neq i \\
  \mathbf{0} & \mbox{otherwise},
\end{array}
\right. \\
A_{3i+2}(a,b) &=\left\{
\begin{array}{ll}
  L_{6i+4+a_i}   & \mbox{if~} a=b \\
  \mathbf{0} & \mbox{otherwise}.
\end{array}
\right.
\end{aligned}
\end{equation}
%where $0$ in the last line denotes the all-zero column vector of length $r$.

\begin{lemma} \label{lm:1q}
If $|\mathbb{F}|\ge 2n+1$, then we can always choose $2n$ distinct elements $\lambda_0,\lambda_1,\dots,\lambda_{2n-1}$ from $\mathbb{F}$ satisfying \eqref{eq:cond}.
\end{lemma}
\begin{proof}
Note that for each value of $i\in\{0,1,\dots,n/3-1\}$, the equation in \eqref{eq:cond} involves $6$ numbers $\lambda_{6i},\lambda_{6i+1},\dots,\lambda_{6i+5}$. As a consequence, each $\lambda_j$ appears in exactly one equation for $0\le j\le 2n-1$. To prove the lemma, we first show that we can choose the $6$ numbers $\lambda_{6i},\lambda_{6i+1},\dots,\lambda_{6i+5}$ from any $7$ distinct elements in $\mathbb{F}$. Indeed, suppose that $\eta_0,\eta_1,\cdots,\eta_6$ are $7$ distinct elements in $\mathbb{F}$. Let $\lambda_{6i+j}=\eta_{j}$ for $j=0,1,2,3,4$. By \eqref{eq:cond}, after fixing the values of $\lambda_{6i},\lambda_{6i+1},\dots,\lambda_{6i+4}$, the last element $\lambda_{6i+5}$ can not be the root of the linear equation $\xi(\lambda_{6i+1}-x)=\lambda_{6i+2}-x$, where 
$$
\xi=\frac{(\lambda_{6i+2}-\lambda_{6i})(\lambda_{6i+2}-\lambda_{6i+4})(\lambda_{6i+1}-\lambda_{6i+3})}{(\lambda_{6i+1}-\lambda_{6i})(\lambda_{6i+1}-\lambda_{6i+4})(\lambda_{6i+2}-\lambda_{6i+3})}\neq 0 .
$$ 
If $\xi=1$, then the linear equation has no roots because $\lambda_{6i+1}\neq \lambda_{6i+2}$. If $\xi\neq 1$, then the linear equation has only one root, so at least one of $\eta_{5}$ and $\eta_{6}$ is not the root of $\xi(\lambda_{6i+1}-x)=\lambda_{6i+2}-x$, and we let $\lambda_{6i+5}$ take that value. This proves that we can choose $6$ numbers $\lambda_{6i},\lambda_{6i+1},\dots,\lambda_{6i+5}$ satisfying \eqref{eq:cond} from any $7$ distinct elements in $\mathbb{F}$. 
Now we let $\mathcal{S}$ be the set of any $7$ distinct elements of $\mathbb{F}$. Then we can choose $\lambda_0,\lambda_1,\cdots,\lambda_5$ satisfying \eqref{eq:cond} from $\mathcal{S}$.
Next we pick another $7$ distinct elements in $\mathbb{F}\backslash\{\lambda_0,\lambda_1,\cdots,\lambda_5\}$ and similarly
we can choose $\lambda_6,\lambda_7,\cdots,\lambda_{11}$ satisfying \eqref{eq:cond} from these $7$ elements.
As  $|\mathbb{F}|\ge 2n+1$, we can repeat this process $n/3$ times, and obtain $2n$ distinct elements $\lambda_0,\lambda_1,\dots,\lambda_{2n-1}$ from $\mathbb{F}$ satisfying \eqref{eq:cond}.
%we the remaining element of $\mathcal{S}_0$ to form the set $\mathcal{S}_1$, and we can choose $\lambda_6,\lambda_7,\cdots,\lambda_{11}$ from the set $\mathcal{S}_1$. From $i=0$ to $n/3-1$, repeat the 
%operation on $\lambda_{6i},\lambda_{6i+1},\dots,\lambda_{6i+5}$ to get all $2n$ elements $\lambda_0,\lambda_1,\cdots,\lambda_{2n-1}$. We can find that $\mathbb{F}\ge 2n+1$ is sufficient.
\end{proof}

\section{A concrete example}\label{example}

We use a concrete example to illustrate the above code construction. We take $n=9$ and $k=5$, so $\ell=2^{9/3}=8$ and $r=n-k=4$. In order to explicitly write out the code construction, we define a sequence of 2D arrays $A_0,A_1,\dots,A_{n-1}$. Each of them has $\ell=8$ columns and $\ell=8$ block rows. The entry at the cross of the $a$th block row and the $b$th column of $A_i$ is $A_i(a,b)$. Note that each entry here is a column vector of length $r=4$. By~\eqref{pcmt} we have
\allowdisplaybreaks
\begin{align*}
& A_0 =
\left[\begin{array}{*{8}{@{\hspace*{0.05in}}c}}
L_0 & L_0-L_1 & \mathbf{0} & \mathbf{0} & \mathbf{0} & \mathbf{0} & \mathbf{0} & \mathbf{0} \\
\mathbf{0} & L_1 & \mathbf{0} & \mathbf{0} & \mathbf{0} & \mathbf{0} & \mathbf{0} & \mathbf{0} \\
\mathbf{0} & \mathbf{0} & L_0 & L_0-L_1 & \mathbf{0} & \mathbf{0} & \mathbf{0} & \mathbf{0} \\
\mathbf{0} & \mathbf{0} & \mathbf{0} & L_1 & \mathbf{0} & \mathbf{0} & \mathbf{0} & \mathbf{0} \\
\mathbf{0} & \mathbf{0} & \mathbf{0} & \mathbf{0} & L_0 & L_0-L_1 & \mathbf{0} & \mathbf{0} \\
\mathbf{0} & \mathbf{0} & \mathbf{0} & \mathbf{0} & \mathbf{0} & L_1 & \mathbf{0} & \mathbf{0} \\
\mathbf{0} & \mathbf{0} & \mathbf{0} & \mathbf{0} & \mathbf{0} & \mathbf{0} & L_0 & L_0-L_1 \\
\mathbf{0} & \mathbf{0} & \mathbf{0} & \mathbf{0} & \mathbf{0} & \mathbf{0} & \mathbf{0} & L_1
\end{array} \hspace*{-0.05in}\right] \\
& A_1 =
\left[\begin{array}{*{8}{@{\hspace*{0.05in}}c}}
L_2 & \mathbf{0} & \mathbf{0} & \mathbf{0} & \mathbf{0} & \mathbf{0} & \mathbf{0} & \mathbf{0} \\
L_3-L_2 & L_3 & \mathbf{0} & \mathbf{0} & \mathbf{0} & \mathbf{0} & \mathbf{0} & \mathbf{0} \\
\mathbf{0} & \mathbf{0} & L_2 & \mathbf{0} & \mathbf{0} & \mathbf{0} & \mathbf{0} & \mathbf{0} \\
\mathbf{0} & \mathbf{0} & L_3-L_2 & L_3 & \mathbf{0} & \mathbf{0} & \mathbf{0} & \mathbf{0} \\
\mathbf{0} & \mathbf{0} & \mathbf{0} & \mathbf{0} & L_2 & \mathbf{0} & \mathbf{0} & \mathbf{0} \\
\mathbf{0} & \mathbf{0} & \mathbf{0} & \mathbf{0} & L_3-L_2 & L_3 & \mathbf{0} & \mathbf{0} \\
\mathbf{0} & \mathbf{0} & \mathbf{0} & \mathbf{0} & \mathbf{0} & \mathbf{0} & L_2 & \mathbf{0} \\
\mathbf{0} & \mathbf{0} & \mathbf{0} & \mathbf{0} & \mathbf{0} & \mathbf{0} & L_3-L_2 & L_3
\end{array} \hspace*{-0.05in}\right] \\
& A_2 =
\left[\begin{array}{*{8}{@{\hspace*{0.05in}}c}}
L_4 & \mathbf{0} & \mathbf{0} & \mathbf{0} & \mathbf{0} & \mathbf{0} & \mathbf{0} & \mathbf{0} \\
\mathbf{0} & L_5 & \mathbf{0} & \mathbf{0} & \mathbf{0} & \mathbf{0} & \mathbf{0} & \mathbf{0} \\
\mathbf{0} & \mathbf{0} & L_4 & \mathbf{0} & \mathbf{0} & \mathbf{0} & \mathbf{0} & \mathbf{0} \\
\mathbf{0} & \mathbf{0} & \mathbf{0} & L_5 & \mathbf{0} & \mathbf{0} & \mathbf{0} & \mathbf{0} \\
\mathbf{0} & \mathbf{0} & \mathbf{0} & \mathbf{0} & L_4 & \mathbf{0} & \mathbf{0} & \mathbf{0} \\
\mathbf{0} & \mathbf{0} & \mathbf{0} & \mathbf{0} & \mathbf{0} & L_5 & \mathbf{0} & \mathbf{0} \\
\mathbf{0} & \mathbf{0} & \mathbf{0} & \mathbf{0} & \mathbf{0} & \mathbf{0} & L_4 & \mathbf{0} \\
\mathbf{0} & \mathbf{0} & \mathbf{0} & \mathbf{0} & \mathbf{0} & \mathbf{0} & \mathbf{0} & L_5
\end{array} \hspace*{-0.05in}\right] \\
& A_3 =
\left[\begin{array}{*{8}{@{\hspace*{0.05in}}c}}
L_6 & \mathbf{0} & L_6-L_7 & \mathbf{0} & \mathbf{0} & \mathbf{0} & \mathbf{0} & \mathbf{0} \\
\mathbf{0} & L_6 & \mathbf{0} & L_6-L_7 & \mathbf{0} & \mathbf{0} & \mathbf{0} & \mathbf{0} \\
\mathbf{0} & \mathbf{0} & L_7 & \mathbf{0} & \mathbf{0} & \mathbf{0} & \mathbf{0} & \mathbf{0} \\
\mathbf{0} & \mathbf{0} & \mathbf{0} & L_7 & \mathbf{0} & \mathbf{0} & \mathbf{0} & \mathbf{0} \\
\mathbf{0} & \mathbf{0} & \mathbf{0} & \mathbf{0} & L_6 & \mathbf{0} & L_6-L_7 & \mathbf{0} \\
\mathbf{0} & \mathbf{0} & \mathbf{0} & \mathbf{0} & \mathbf{0} & L_6 & \mathbf{0} & L_6-L_7 \\
\mathbf{0} & \mathbf{0} & \mathbf{0} & \mathbf{0} & \mathbf{0} & \mathbf{0} & L_7 & \mathbf{0} \\
\mathbf{0} & \mathbf{0} & \mathbf{0} & \mathbf{0} & \mathbf{0} & \mathbf{0} & \mathbf{0} & L_7
\end{array} \hspace*{-0.05in}\right] \\
& A_4 =
\left[\begin{array}{*{8}{@{\hspace*{0.05in}}c}}
L_8 & \mathbf{0} & \mathbf{0} & \mathbf{0} & \mathbf{0} & \mathbf{0} & \mathbf{0} & \mathbf{0} \\
\mathbf{0} & L_8 & \mathbf{0} & \mathbf{0} & \mathbf{0} & \mathbf{0} & \mathbf{0} & \mathbf{0} \\
L_9-L_8 & \mathbf{0} & L_9 & \mathbf{0} & \mathbf{0} & \mathbf{0} & \mathbf{0} & \mathbf{0} \\
\mathbf{0} & L_9-L_8 & \mathbf{0} & L_9 & \mathbf{0} & \mathbf{0} & \mathbf{0} & \mathbf{0} \\
\mathbf{0} & \mathbf{0} & \mathbf{0} & \mathbf{0} & L_8 & \mathbf{0} & \mathbf{0} & \mathbf{0} \\
\mathbf{0} & \mathbf{0} & \mathbf{0} & \mathbf{0} & \mathbf{0} & L_8 & \mathbf{0} & \mathbf{0} \\
\mathbf{0} & \mathbf{0} & \mathbf{0} & \mathbf{0} & L_9-L_8 & \mathbf{0} & L_9 & \mathbf{0} \\
\mathbf{0} & \mathbf{0} & \mathbf{0} & \mathbf{0} & \mathbf{0} & L_9-L_8 & \mathbf{0} & L_9
\end{array} \hspace*{-0.05in}\right] \\
& A_5 =
\left[\begin{array}{*{8}{@{\hspace*{0.05in}}c}}
L_{10} & \mathbf{0} & \mathbf{0} & \mathbf{0} & \mathbf{0} & \mathbf{0} & \mathbf{0} & \mathbf{0} \\
\mathbf{0} & L_{10} & \mathbf{0} & \mathbf{0} & \mathbf{0} & \mathbf{0} & \mathbf{0} & \mathbf{0} \\
\mathbf{0} & \mathbf{0} & L_{11} & \mathbf{0} & \mathbf{0} & \mathbf{0} & \mathbf{0} & \mathbf{0} \\
\mathbf{0} & \mathbf{0} & \mathbf{0} & L_{11} & \mathbf{0} & \mathbf{0} & \mathbf{0} & \mathbf{0} \\
\mathbf{0} & \mathbf{0} & \mathbf{0} & \mathbf{0} & L_{10} & \mathbf{0} & \mathbf{0} & \mathbf{0} \\
\mathbf{0} & \mathbf{0} & \mathbf{0} & \mathbf{0} & \mathbf{0} & L_{10} & \mathbf{0} & \mathbf{0} \\
\mathbf{0} & \mathbf{0} & \mathbf{0} & \mathbf{0} & \mathbf{0} & \mathbf{0} & L_{11} & \mathbf{0} \\
\mathbf{0} & \mathbf{0} & \mathbf{0} & \mathbf{0} & \mathbf{0} & \mathbf{0} & \mathbf{0} & L_{11}
\end{array} \hspace*{-0.05in}\right] \\
& {\scriptsize A_6 =
\left[\begin{array}{*{8}{@{\hspace*{0.03in}}c}}
L_{12} & \mathbf{0} & \mathbf{0} & \mathbf{0} & L_{12}-L_{13} & \mathbf{0} & \mathbf{0} & \mathbf{0} \\
\mathbf{0} & L_{12} & \mathbf{0} & \mathbf{0} & \mathbf{0} & L_{12}-L_{13} & \mathbf{0} & \mathbf{0} \\
\mathbf{0} & \mathbf{0} & L_{12} & \mathbf{0} & \mathbf{0} & \mathbf{0} & L_{12}-L_{13} & \mathbf{0} \\
\mathbf{0} & \mathbf{0} & \mathbf{0} & L_{12} & \mathbf{0} & \mathbf{0} & \mathbf{0} & L_{12}-L_{13} \\
\mathbf{0} & \mathbf{0} & \mathbf{0} & \mathbf{0} & L_{13} & \mathbf{0} & \mathbf{0} & \mathbf{0} \\
\mathbf{0} & \mathbf{0} & \mathbf{0} & \mathbf{0} & \mathbf{0} & L_{13} & \mathbf{0} & \mathbf{0} \\
\mathbf{0} & \mathbf{0} & \mathbf{0} & \mathbf{0} & \mathbf{0} & \mathbf{0} & L_{13} & \mathbf{0} \\
\mathbf{0} & \mathbf{0} & \mathbf{0} & \mathbf{0} & \mathbf{0} & \mathbf{0} & \mathbf{0} & L_{13}
\end{array} \hspace*{-0.05in}\right] }\\
& {\scriptsize A_7 =
\left[\begin{array}{*{8}{@{\hspace*{0.03in}}c}}
L_{14} & \mathbf{0} & \mathbf{0} & \mathbf{0} & \mathbf{0} & \mathbf{0} & \mathbf{0} & \mathbf{0} \\
\mathbf{0} & L_{14} & \mathbf{0} & \mathbf{0} & \mathbf{0} & \mathbf{0} & \mathbf{0} & \mathbf{0} \\
\mathbf{0} & \mathbf{0} & L_{14} & \mathbf{0} & \mathbf{0} & \mathbf{0} & \mathbf{0} & \mathbf{0} \\
\mathbf{0} & \mathbf{0} & \mathbf{0} & L_{14} & \mathbf{0} & \mathbf{0} & \mathbf{0} & \mathbf{0} \\
L_{15}-L_{14} & \mathbf{0} & \mathbf{0} & \mathbf{0} & L_{15} & \mathbf{0} & \mathbf{0} & \mathbf{0} \\
\mathbf{0} & L_{15}-L_{14} & \mathbf{0} & \mathbf{0} & \mathbf{0} & L_{15} & \mathbf{0} & \mathbf{0} \\
\mathbf{0} & \mathbf{0} & L_{15}-L_{14} & \mathbf{0} & \mathbf{0} & \mathbf{0} & L_{15} & \mathbf{0} \\
\mathbf{0} & \mathbf{0} & \mathbf{0} & L_{15}-L_{14} & \mathbf{0} & \mathbf{0} & \mathbf{0} & L_{15}
\end{array} \hspace*{-0.05in}\right]} \\
& A_8 =
\left[\begin{array}{*{8}{@{\hspace*{0.05in}}c}}
L_{16} & \mathbf{0} & \mathbf{0} & \mathbf{0} & \mathbf{0} & \mathbf{0} & \mathbf{0} & \mathbf{0} \\
\mathbf{0} & L_{16} & \mathbf{0} & \mathbf{0} & \mathbf{0} & \mathbf{0} & \mathbf{0} & \mathbf{0} \\
\mathbf{0} & \mathbf{0} & L_{16} & \mathbf{0} & \mathbf{0} & \mathbf{0} & \mathbf{0} & \mathbf{0} \\
\mathbf{0} & \mathbf{0} & \mathbf{0} & L_{16} & \mathbf{0} & \mathbf{0} & \mathbf{0} & \mathbf{0} \\
\mathbf{0} & \mathbf{0} & \mathbf{0} & \mathbf{0} & L_{17} & \mathbf{0} & \mathbf{0} & \mathbf{0} \\
\mathbf{0} & \mathbf{0} & \mathbf{0} & \mathbf{0} & \mathbf{0} & L_{17} & \mathbf{0} & \mathbf{0} \\
\mathbf{0} & \mathbf{0} & \mathbf{0} & \mathbf{0} & \mathbf{0} & \mathbf{0} & L_{17} & \mathbf{0} \\
\mathbf{0} & \mathbf{0} & \mathbf{0} & \mathbf{0} & \mathbf{0} & \mathbf{0} & \mathbf{0} & L_{17}
\end{array} \hspace*{-0.05in}\right].
\end{align*}

\subsection{MDS property} \label{subsect:MDS}
We first explain why this code is an MDS array code. In order to prove the MDS property, we need to show that any $4$ nodes of this code can be recovered from the remaining $5$ nodes. We will use two concrete cases to illustrate how to prove the MDS property. 

{\bf First case: How to recover $C_0,C_1,C_3,C_5$ from the remaining $5$ nodes.} To that end, we only need to show that if all the coordinates of $C_2,C_4,C_6,C_7,C_8$ are $0$, then $C_0=C_1=C_3=C_5=\mathbf{0}$ is the only solution to the parity check equations \eqref{eq:pc}. Indeed, when $C_2=C_4=C_6=C_7=C_8=\mathbf{0}$, the equations \eqref{eq:pc} can be written as follows: for $a=0,1,\cdots,7$,
\begin{align}\label{eq:pc_case1}
    &A_0(a,0)C_0(0)+\cdots+A_0(a,7)C_0(7)\\\nonumber
    +&A_1(a,0)C_1(0)+\cdots+A_1(a,7)C_1(7)\\\nonumber
    +&A_3(a,0)C_3(0)+\cdots+A_3(a,7)C_3(7)\\\nonumber
    +&A_5(a,0)C_5(0)+\cdots+A_5(a,7)C_5(7)=\mathbf{0}.
\end{align}

Now we only consider $a=0,1,2,3$. Since $A_i(a,b)=\mathbf{0}$ for every $i=0,1,3,5$ and $a=0,1,2,3$, $b=4,5,6,7$, the equation can be further written as 
$$MC=(A_0',A_1',A_3',A_5')\cdot C=\mathbf{0},$$
where 
\begin{align*}
    &A_0'=\begin{bmatrix}
        L_0 & L_0-L_1 & \mathbf{0} & \mathbf{0}\\
        \mathbf{0} & L_1 & \mathbf{0} & \mathbf{0}\\
        \mathbf{0} & \mathbf{0} & L_0 & L_0-L_1\\
        \mathbf{0} & \mathbf{0} & \mathbf{0} & L_1
    \end{bmatrix}, \\
    &A_1'=\begin{bmatrix}
        L_2 & \mathbf{0} & \mathbf{0} & \mathbf{0}\\
        L_3-L_2 & L_3 & \mathbf{0} & \mathbf{0}\\
        \mathbf{0} & \mathbf{0} & L_2 & \mathbf{0}\\
        \mathbf{0} & \mathbf{0} & L_3-L_2 & L_3
    \end{bmatrix}, \\
    &A_3'=\begin{bmatrix}
        L_6 & \mathbf{0} & L_6-L_7 & \mathbf{0}\\
        \mathbf{0} & L_6 & \mathbf{0} & L_6-L_7\\
        \mathbf{0} & \mathbf{0} & L_7 & \mathbf{0}\\
        \mathbf{0} & \mathbf{0} & \mathbf{0} & L_7
    \end{bmatrix}, \\
    &A_5'=\begin{bmatrix}
        L_{10} & \mathbf{0} & \mathbf{0} & \mathbf{0}\\
        \mathbf{0} & L_{10} & \mathbf{0} & \mathbf{0}\\
        \mathbf{0} & \mathbf{0} & L_{11} & \mathbf{0}\\
        \mathbf{0} & \mathbf{0} & \mathbf{0} & L_{11}
    \end{bmatrix},
\end{align*}
and
\begin{align*}
     C=&\Big(C_0(0),C_0(1),C_0(2),C_0(3),C_1(0),C_1(1),C_1(2),C_1(3),\\
    &C_3(0),C_3(1),C_3(2),C_3(3),C_5(0),C_5(1),C_5(2),C_5(3)\Big)^T.
\end{align*}

Summing up the equations of block row\footnote{In this paper, we assume that the row index and column index start from $0$.} $2$ and $3$ in $MC=\mathbf{0}$, we obtain 
\begin{align*}
	&L_0(C_0(2) +C_0(3)) + L_3(C_1(2)+C_1(3)) \\
	+&L_7(C_3(2)+C_3(3)) + L_{11}(C_5(2)+C_5(3)) = \mathbf{0}.
\end{align*}
Since $[L_0, L_3, L_7, L_{11}]$ is a $4\times 4$ Vandermonde matrix, we have $C_0(2)+C_0(3)=C_1(2)+C_1(3)=C_3(2)+C_3(3)=C_5(2)+C_5(3)=0$. Taking these back into the equations of block row $2$ and $3$ in $MC=0$, we obtain
\begin{align*}
& -L_1 C_0(3) + L_2 C_1(2)  +L_7 C_3(2)+ L_{11}C_5(2) =\mathbf{0} , \\
& L_1 C_0(3) -L_2 C_1(2) + L_7 C_3(3)  +L_{11}C_5(3) =\mathbf{0} .
\end{align*}
Similarly we have $C_0(3)=C_1(2)=C_3(2)=C_3(3)=C_5(2)=C_5(3)=0$. 
Combining this with $C_0(2)+C_0(3)=C_1(2)+C_1(3)=0$, we obtain that $C_i(j)=0$ for all $i=0,1,3,5$ and $j=2,3$. Taking this into block row $0$ and $1$ of $MC=\mathbf{0}$, we obtain
\begin{equation} \label{eq:bk}
\begin{aligned}
&L_0 C_0(0) +(L_0  -L_1)C_0(1)\\ + &L_2 C_1(0) +L_6 C_3(0) + L_{10}C_5(0) = \mathbf{0},  \\
&L_1 C_0(1) +(L_3  -L_2) C_1(0)\\ + &L_3 C_1(1) + L_6 C_3(1)  +L_{10}C_5(1) = \mathbf{0} .
\end{aligned}
\end{equation}
Summing up these two equations, we have
\begin{align*}
&L_0(C_0(0) +C_0(1)) +L_3(C_1(0)+C_1(1)) \\
+&L_6 (C_3(0)+C_3(1)) +L_{10} (C_5(0)+C_5(1)) =\mathbf{0}.
\end{align*}
Therefore, $C_0(0) +C_0(1)=C_1(0)+C_1(1)=C_3(0)+C_3(1)=C_5(0)+C_5(1)=0$. Taking these back into \eqref{eq:bk}, we obtain
\begin{align*}
& -L_1 C_0(1) + L_2 C_1(0) 
 +L_6 C_3(0) + L_{10}C_5(0) = \mathbf{0},  \\
& L_1 C_0(1) -L_2 C_1(0) + L_6 C_3(1)  +L_{10}C_5(1) = \mathbf{0} .
\end{align*}
Therefore, $C_0(1)=C_1(0)=C_3(0)=C_3(1)=C_5(0)=C_5(1)=0$. Combining these with $C_0(0) +C_0(1)=C_1(0)+C_1(1)=0$, we have $C_i(j)=0$ for all $i=0,1,3,5$ and $j=0,1$.

Now we have used $a=0,1,2,3$ of equations \eqref{eq:pc_case1} to conclude that $C_i(j)=0$ for all $i=0,1,3,5$ and $j=0,1,2,3$. The same analysis of $a=4,5,6,7$ of equations \eqref{eq:pc_case1} allows us to obtain that $C_i(j)=0$ for all $i=0,1,3,5$ and $j=4,5,6,7$. Thus we have shown that $C_0=C_1=C_3=C_5=\mathbf{0}$. This proves that $C_0,C_1,C_3,C_5$ can be recovered from the remaining $5$ nodes.

{\bf Second case: How to recover $C_0,C_1,C_2,C_5$ from the remaining $5$ nodes.} To that end, we only need to show that if all the coordinates of $C_3,C_4,C_6,C_7,C_8$ are $\mathbf{0}$, then $C_0=C_1=C_2=C_5=\mathbf{0}$ is the only solution to the parity check equations \eqref{eq:pc}. Indeed, when $C_3=C_4=C_6=C_7=C_8=\mathbf{0}$, the equations \eqref{eq:pc} can be written as follows: for $a=0,1,\cdots,7$, we have
\begin{align}\label{eq:pc_case2}
    &A_0(a,0)C_0(0)+\cdots+A_0(a,7)C_0(7)\\\nonumber
    +&A_1(a,0)C_1(0)+\cdots+A_1(a,7)C_1(7)\\\nonumber
    +&A_2(a,0)C_2(0)+\cdots+A_2(a,7)C_2(7)\\\nonumber
    +&A_5(a,0)C_5(0)+\cdots+A_5(a,7)C_5(7)=\mathbf{0}.
\end{align}

Now we only consider $a=0,1,2,3$. Since $A_i(a,b)=\mathbf{0}$ for every $i=0,1,2,5$ and $a=0,1,2,3$, $b=4,5,6,7$, the equation can be further written as 
$$MC=(A_0',A_1',A_2',A_5')\cdot C=\mathbf{0},$$
where 
\begin{align*}
    &A_0'=\begin{bmatrix}
        L_0 & L_0-L_1 & \mathbf{0} & \mathbf{0}\\
        \mathbf{0} & L_1 & \mathbf{0} & \mathbf{0}\\
        \mathbf{0} & \mathbf{0} & L_0 & L_0-L_1\\
        \mathbf{0} & \mathbf{0} & \mathbf{0} & L_1
    \end{bmatrix}, \\
    &A_1'=\begin{bmatrix}
        L_2 & \mathbf{0} & \mathbf{0} & \mathbf{0}\\
        L_3-L_2 & L_3 & \mathbf{0} & \mathbf{0}\\
        \mathbf{0} & \mathbf{0} & L_2 & \mathbf{0}\\
        \mathbf{0} & \mathbf{0} & L_3-L_2 & L_3
    \end{bmatrix}, \\
    &A_2'=\begin{bmatrix}
        L_4 & \mathbf{0} & \mathbf{0} & \mathbf{0}\\
        \mathbf{0} & L_5 & \mathbf{0} & \mathbf{0}\\
        \mathbf{0} & \mathbf{0} & L_4 & \mathbf{0}\\
        \mathbf{0} & \mathbf{0} & \mathbf{0} & L_5
    \end{bmatrix}, \\
    &A_5'=\begin{bmatrix}
        L_{10} & \mathbf{0} & \mathbf{0} & \mathbf{0}\\
        \mathbf{0} & L_{10} & \mathbf{0} & \mathbf{0}\\
        \mathbf{0} & \mathbf{0} & L_{11} & \mathbf{0}\\
        \mathbf{0} & \mathbf{0} & \mathbf{0} & L_{11}
    \end{bmatrix},
\end{align*}
and
\begin{align*}
    C=&\Big(C_0(0),C_0(1),C_0(2),C_0(3),C_1(0),C_1(1),C_1(2),C_1(3),\\
    &C_2(0),C_2(1),C_2(2),C_2(3),C_5(0),C_5(1),C_5(2),C_5(3)\Big)^T.
\end{align*}

The equations of block row $0$ and $1$ in $MC=\mathbf{0}$ can be written in the matrix form:
$$M'C'=(\mathcal{A}_0,\mathcal{A}_1,\mathcal{A}_2,\mathcal{A}_5)\cdot C'=\mathbf{0}, $$
where  
\begin{align*}
	&\mathcal{A}_0=\begin{bmatrix}
		L_0 & L_0-L_1 \\
		\mathbf{0} & L_1  
	\end{bmatrix},\quad
	\mathcal{A}_1=\begin{bmatrix}
		L_2 & \mathbf{0} \\
		L_3-L_2 & L_3 
	\end{bmatrix},\\
	&\mathcal{A}_2=\begin{bmatrix}
		L_4 & \mathbf{0} \\
		\mathbf{0} & L_5  
	\end{bmatrix},\quad
	\mathcal{A}_5=\begin{bmatrix}
		L_{10} & \mathbf{0} \\
		\mathbf{0} & L_{10} 
	\end{bmatrix},
\end{align*}
{\small \begin{align*}
		C'=\Big( C_0(0),C_0(1),C_1(0),C_1(1),C_2(0),C_2(1),C_5(0),C_5(1)\Big)^T.
	\end{align*}} 

Next we show that $M'$ is invertible, so $C_i(j)=0$ for all $i=0,1,2,5$ and all $j=0,1$. To that end, we perform the following elementary column operations on $M'$:
\begin{itemize}
    \item add column $0$ multiplied by $(-1)$ to column $1$ of $\mathcal{A}_0$, denoted as $\mathcal{A}_0'$,
    \item add column $1$ multiplied by $(-1)$ to column $0$ of $\mathcal{A}_1$, denoted as $\mathcal{A}_1'$,
    \item $\mathcal{A}_2'=\mathcal{A}_2$,
    \item $\mathcal{A}_5'=\mathcal{A}_5$.
\end{itemize}
Then we have
\begin{align*}
	&\mathcal{A}'_0=\begin{bmatrix}
		L_0 & -L_1 \\
		\mathbf{0} & L_1  
	\end{bmatrix},
	\mathcal{A}'_1=\begin{bmatrix}
		L_2 & \mathbf{0} \\
		-L_2 & L_3 
	\end{bmatrix},\\
	&\mathcal{A}'_2=\begin{bmatrix}
		L_4 & \mathbf{0} \\
		\mathbf{0} & L_5  
	\end{bmatrix},
	\mathcal{A}'_5=\begin{bmatrix}
		L_{10} & \mathbf{0} \\
		\mathbf{0} & L_{10} 
	\end{bmatrix}.
\end{align*}

Let \begin{align*}
	M''=&(\mathcal{A}_0',\mathcal{A}_1',\mathcal{A}_2',\mathcal{A}_5')\\
	=&\begin{bmatrix}
		L_0&-L_1&L_2&\mathbf{0}&L_4&\mathbf{0}&L_{10}&\mathbf{0}\\
		\mathbf{0}&L_1&-L_2&L_3&\mathbf{0}&L_5&\mathbf{0}&L_{10}
	\end{bmatrix}.
\end{align*} The $8\times 8$ matrix $M''$ can be written in the following form: 
$$
\left[\begin{array}{cccccccc}
1 & -1 & 1 & 0 & 1 & 0 & 1 & 0 \\
\lambda_0 & -\lambda_1 & \lambda_2 & 0 & \lambda_4 & 0 & \lambda_{10} & 0 \\
\lambda_0^2 & -\lambda_1^2 & \lambda_2^2 & 0 & \lambda_4^2 & 0 & \lambda_{10}^2 & 0 \\
\lambda_0^3 & -\lambda_1^3 & \lambda_2^3 & 0 & \lambda_4^3 & 0 & \lambda_{10}^3 & 0 \\
0 & 1 & -1 & 1 & 0 & 1 & 0 & 1 \\
0 & \lambda_1 & -\lambda_2 & \lambda_3 & 0 & \lambda_5 & 0 & \lambda_{10} \\
0 & \lambda_1^2 & -\lambda_2^2 & \lambda_3^2 & 0 & \lambda_5^2 & 0 & \lambda_{10}^2 \\
0 & \lambda_1^3 & -\lambda_2^3 & \lambda_3^3 & 0 & \lambda_5^3 & 0 & \lambda_{10}^3
\end{array}\right] .
$$

It is obvious that $M''$ is invertible if and only if $M'$ is invertible because the elementary column operations do not change the invertibility of matrix.  

Now we only need to show that $M''$ is invertible, so is $M'$, then by $M'C'=\mathbf{0}$, we have $C'=\mathbf{0}$. 
%that is $C_i(j)=0$ for all $i=0,1,2,5$ and all $j=0,1$. 
To that end, let us consider a row vector $f=(f_{0,0},f_{0,1},f_{0,2},f_{0,3},f_{1,0},f_{1,1},f_{1,2},f_{1,3})$ of length $8$ and a column vector $\vec{y}=(y_{0,0},y_{0,1},y_{1,0},y_{1,1},y_{2,0},y_{2,1},y_{5,0},y_{5,1})^T$ of length $8$. The matrix $M''$ is invertible if and only if $\vec{y}=\mathbf{0}$ is the only solution to the equation $M''\vec{y}=\mathbf{0}$. Define two polynomials 
\begin{align*}
		f_0(x)=&(x-\lambda_0)(x-\lambda_4)(x-\lambda_{10})\\=&f_{0,0}+f_{0,1}x+f_{0,2}x^2+f_{0,3}x^3,\\ f_1(x)=&(x-\lambda_3)(x-\lambda_5)(x-\lambda_{10})\\=&f_{1,0}+f_{1,1}x+f_{1,2}x^2+f_{1,3}x^3.
\end{align*}
Then $fM''\vec{y}=\mathbf{0}$ implies that
\begin{align*}
\begin{bmatrix}
0 & -f_0(\lambda_1) & f_0(\lambda_2) & 0 & 0 & 0 & 0 & 0 \\
0 & f_1(\lambda_1) & -f_1(\lambda_2) & 0 & 0 & 0 & 0 & 0 
\end{bmatrix}
\begin{bmatrix}
    y_{0,0}\\y_{0,1}\\y_{1,0}\\y_{1,1}\\y_{2,0}\\y_{2,1}\\y_{5,0}\\y_{5,1}
\end{bmatrix}
=\mathbf{0}.
\end{align*}

The equations can be written as $\begin{bmatrix}
    -f_0(\lambda_1)&f_0(\lambda_2)\\
    f_1(\lambda_1)&-f_1(\lambda_2)
\end{bmatrix}
\begin{bmatrix}
    y_{0,1}\\y_{1,0}
\end{bmatrix}=\mathbf{0}$. 
By the fact that $\lambda_1, \lambda_2,\lambda_{10}$ are distinct and condition \eqref{eq:cond}, the determinant of the left matrix is $f_0(\lambda_1)f_1(\lambda_2)-f_0(\lambda_2)f_1(\lambda_1)=(\lambda_1-\lambda_{10})(\lambda_2-\lambda_{10})\Big((\lambda_1-\lambda_0)(\lambda_1-\lambda_4)(\lambda_2-\lambda_3)(\lambda_2-\lambda_5)-(\lambda_2-\lambda_0)(\lambda_2-\lambda_4)(\lambda_1-\lambda_3)(\lambda_1-\lambda_5)\Big)\neq 0$, so $y_{0,1}=y_{1,0}=0$. Taking these back into $M''\vec{y}=\mathbf{0}$, we obtain 
$$
\begin{bmatrix}
    L_0 & \mathbf{0} & L_4 & \mathbf{0} & L_{10} & \mathbf{0}\\
    \mathbf{0} & L_3 & \mathbf{0} & L_5 & \mathbf{0} & L_{10}
\end{bmatrix}
\begin{bmatrix}
    y_{0,0}\\y_{1,1}\\y_{2,0}\\y_{2,1}\\y_{5,0}\\y_{5,1}
\end{bmatrix}=\mathbf{0}.
$$

Since the first three rows of the $4\times 3$ matrix $[L_0,L_4,L_{10}]$ is a $3\times 3$ Vandermonde matrix, we have $y_{0,0}=y_{2,0}=y_{5,0}=0$. Similarly the first three rows of the $4\times 3$ matrix $[L_3,L_5,L_{10}]$ is a $3\times 3$ Vandermonde matrix, we have $y_{1,1}=y_{2,1}=y_{5,1}=0$. Therefore, $\vec{y}=\mathbf{0}$, we conclude that $M''$ is invertible. From the previous analysis, $M''$ is invertible implies $C'=\mathbf{0}$, that is $C_i(j)=0$ for all $i=0,1,2,5$ and all $j=0,1$. Using exactly the same method, one can show that $C_i(j)=0$ for all $i=0,1,2,5$ and all $j=2,3$ from the equations of block row $2$ and $3$ in $MC=\mathbf{0}$.

Now we have used $a=0,1,2,3$ of equations \eqref{eq:pc_case2} to conclude that $C_i(j)=0$ for all $i=0,1,2,5$ and $j=0,1,2,3$. The same analysis of $a=4,5,6,7$ of equations \eqref{eq:pc_case2} allows us to obtain that $C_i(j)=0$ for all $i=0,1,2,5$ and $j=4,5,6,7$. Thus we have shown that $C_0=C_1=C_2=C_5=\mathbf{0}$. This proves that $C_0,C_1,C_2,C_5$ can be recovered from the remaining $5$ nodes.

\subsection{Optimal repair bandwidth for single node failure}
We also use two cases to illustrate the repair procedure.

{\bf First case: How to repair $C_0$.} Note that the parity check equations in \eqref{eq:pc} can be written in the matrix form 
\begin{equation} \label{eq:repairexample}
A_0 C_0+A_1 C_1+A_2 C_2+\dots+A_8 C_8=\mathbf{0} ,
\end{equation}
where each $C_i$ is a column vector of length $\ell=8$. Each block row in the matrices $A_0,\dots,A_8$ corresponds to a set of $r=n-k=4$ parity check equations because the length of each $L_i$ is $4$. Since there are $8$ block rows in each matrix $A_i$, we have $8$ sets of parity check equations in total. The repair of $C_0$ only involves $4$ out of these $8$ sets of parity check equations. More precisely, among the $8$ block rows in each matrix $A_i$, we only need to look at the block rows whose indices lie in the set $\{0,2,4,6\}$. 
These $4$ block rows of parity check equations can again be organized in the matrix form
\begin{equation} \label{eq:pceq0}
\widetilde{A}_0 \widetilde{C}_0 + \widehat{A}_0 \widehat{C}_0
+ \sum_{i=1}^8 \overline{A}_i \overline{C}_i = \mathbf{0} ,
\end{equation}
where $\widetilde{A}_0, \widehat{A}_0, \overline{A}_i, 1\le i\le 8$ are all $4\times 4$ matrices, and $\widetilde{C}_0,\widehat{C}_0,\overline{C}_i, 1\le i\le 8$ are all column vectors of length $4$. More specifically, the matrices in \eqref{eq:pceq0} are
{\scriptsize
\begin{align*}
& \widetilde{A}_0 =
\left[\begin{array}{*{4}{@{\hspace*{0.05in}}c}}
L_0 & \mathbf{0} & \mathbf{0} & \mathbf{0} \\
\mathbf{0} & L_0 & \mathbf{0} & \mathbf{0} \\
\mathbf{0} & \mathbf{0} & L_0 & \mathbf{0} \\
\mathbf{0} & \mathbf{0} & \mathbf{0} & L_0 \\
\end{array} \hspace*{-0.05in}\right]\ 
\widehat{A}_0 =
\left[\begin{array}{*{4}{@{\hspace*{0.05in}}c}}
-L_1 & \mathbf{0} & \mathbf{0} & \mathbf{0} \\
\mathbf{0} & -L_1 & \mathbf{0} & \mathbf{0} \\
\mathbf{0} & \mathbf{0} & -L_1 & \mathbf{0} \\
\mathbf{0} & \mathbf{0} & \mathbf{0} & -L_1 \\
\end{array} \hspace*{-0.05in}\right] \\
& \overline{A}_1 =
\left[\begin{array}{*{4}{@{\hspace*{0.05in}}c}}
L_2 & \mathbf{0} & \mathbf{0} & \mathbf{0} \\
\mathbf{0} & L_2 & \mathbf{0} & \mathbf{0} \\
\mathbf{0} & \mathbf{0} & L_2 & \mathbf{0} \\
\mathbf{0} & \mathbf{0} & \mathbf{0} & L_2 \\
\end{array} \hspace*{-0.05in}\right] \ 
\overline{A}_2 =
\left[\begin{array}{*{4}{@{\hspace*{0.05in}}c}}
L_4 & \mathbf{0} & \mathbf{0} & \mathbf{0} \\
\mathbf{0} & L_4 & \mathbf{0} & \mathbf{0} \\
\mathbf{0} & \mathbf{0} & L_4 & \mathbf{0} \\
\mathbf{0} & \mathbf{0} & \mathbf{0} & L_4 
\end{array} \hspace*{-0.05in}\right] \\
& \overline{A}_3 =
\left[\begin{array}{*{4}{@{\hspace*{0.05in}}c}}
L_6 & L_6-L_7 & \mathbf{0} & \mathbf{0} \\
\mathbf{0} & L_7 & \mathbf{0} & \mathbf{0} \\
\mathbf{0} & \mathbf{0} & L_6 & L_6-L_7 \\
\mathbf{0} & \mathbf{0} & \mathbf{0} & L_7  
\end{array} \hspace*{-0.05in}\right]\ 
\overline{A}_4 =
\left[\begin{array}{*{4}{@{\hspace*{0.05in}}c}}
L_8 & \mathbf{0} & \mathbf{0} & \mathbf{0} \\
L_9-L_8 & L_9 & \mathbf{0} & \mathbf{0} \\
\mathbf{0} & \mathbf{0} & L_8 & \mathbf{0} \\
\mathbf{0} & \mathbf{0} & L_9-L_8 & L_9 
\end{array} \hspace*{-0.05in}\right] \\
&\overline{A}_5 =
\left[\begin{array}{*{4}{@{\hspace*{0.05in}}c}}
L_{10} & \mathbf{0} & \mathbf{0} & \mathbf{0} \\
\mathbf{0} & L_{11} & \mathbf{0} & \mathbf{0} \\
\mathbf{0} & \mathbf{0} & L_{10} & \mathbf{0} \\
\mathbf{0} & \mathbf{0} & \mathbf{0} & L_{11}  
\end{array} \hspace*{-0.05in}\right] \  
\overline{A}_6 =
\left[\begin{array}{*{4}{@{\hspace*{0.03in}}c}}
L_{12} & \mathbf{0} & L_{12}-L_{13} & \mathbf{0} \\
\mathbf{0} & L_{12} & \mathbf{0} & L_{12}-L_{13}  \\
\mathbf{0} & \mathbf{0} & L_{13} & \mathbf{0} \\
\mathbf{0} & \mathbf{0} & \mathbf{0} & L_{13} 
\end{array} \hspace*{-0.05in}\right] \\
& \overline{A}_7 =
\left[\begin{array}{*{4}{@{\hspace*{0.03in}}c}}
L_{14} & \mathbf{0} & \mathbf{0} & \mathbf{0}  \\
\mathbf{0} & L_{14} & \mathbf{0} & \mathbf{0} \\
L_{15}-L_{14} & \mathbf{0} & L_{15} & \mathbf{0} \\
\mathbf{0} & L_{15}-L_{14} & \mathbf{0} & L_{15} 
\end{array} \hspace*{-0.05in}\right]\ 
\overline{A}_8 =
\left[\begin{array}{*{4}{@{\hspace*{0.05in}}c}}
L_{16} & \mathbf{0} & \mathbf{0} & \mathbf{0} \\
\mathbf{0} & L_{16} & \mathbf{0} & \mathbf{0} \\
\mathbf{0} & \mathbf{0} & L_{17} & \mathbf{0} \\
\mathbf{0} & \mathbf{0} & \mathbf{0} & L_{17} 
\end{array} \hspace*{-0.05in}\right]
\end{align*}
}%
and the column vectors in \eqref{eq:pceq0} are
\begin{align*}
& \widetilde{C}_0 =
\left[ \begin{array}{c}
  C_0(0)+C_0(1) \\
  C_0(2)+C_0(3) \\
  C_0(4)+C_0(5) \\
  C_0(6)+C_0(7) 
\end{array}
\right] , \quad
\widehat{C}_0 = 
\left[ \begin{array}{c}
  C_0(1) \\
  C_0(3) \\
  C_0(5) \\
  C_0(7) 
\end{array}
\right] , 
\\
& \overline{C}_i =
\left[ \begin{array}{c}
  C_i(0) \\
  C_i(2) \\
  C_i(4) \\
  C_i(6) 
\end{array}
\right]
\text{~for~} 1\le i\le 8 .
\end{align*}
Here we make an important observation: The matrices $\overline{A}_3,\overline{A}_4,\dots,\overline{A}_8$ are precisely the $6$ parity check matrices that would appear in our MSR code construction with code length $n=6$ and subpacketization $\ell=2^{6/3}=4$. The other $4$ matrices $\widetilde{A}_0,\widehat{A}_0,\overline{A}_1,\overline{A}_2$ are all block-diagonal matrices\footnote{They are block-diagonal matrices because every entry in these matrices is a column vector of length $4$.}, and the diagonal entries are the same within each matrix. Moreover, the $\lambda_i$'s (or equivalently $L_i$'s) that appear in $\widetilde{A}_0,\widehat{A}_0,\overline{A}_1,\overline{A}_2$ do not intersect with the $\lambda_i$'s that appear in $\overline{A}_3,\overline{A}_4,\dots,\overline{A}_8$.
The method we used to prove the MDS property of our MSR code construction in Section~\ref{subsect:MDS} can be easily generalized to show that \eqref{eq:pceq0} also defines an MDS array code $(\widetilde{C}_0,\widehat{C}_0,\overline{C}_1,\overline{C}_2,\dots,\overline{C}_8)$ with code length $10$ and code dimension $6$. Therefore, $\widetilde{C}_0$ and $\widehat{C}_0$ can be recovered from any $6$ vectors in the set $\{\overline{C}_1,\overline{C}_2,\dots,\overline{C}_8\}$. Once we know the values of $\widetilde{C}_0$ and $\widehat{C}_0$, we are able to recover all the coordinates of $C_0$.

{\bf Second case: How to repair $C_8$.}
In order to repair $C_8$, we sum up four pairs of block rows of each matrix $A_i$ in equation \eqref{eq:repairexample}. Specifically, we sum up $0$th block row and the $4$th block row, the $1$st block row and the $5$th block row, the $2$nd block row and the $6$th block row, the $3$rd block row and the $7$th block row. (The row index starts from $0$.) In this way, we obtain $4$ block rows of parity check equations from the original $8$ block rows of parity check equations in \eqref{eq:repairexample}. These $4$ block rows of parity check equations can be written in the matrix form
\begin{equation} \label{eq:pceq8}
\widetilde{A}_8 \widetilde{C}_8 + \widehat{A}_8 \widehat{C}_8
+ \sum_{i=0}^7 \overline{A}_i \overline{C}_i = \mathbf{0} ,
\end{equation}
where $\widetilde{A}_8, \widehat{A}_8, \overline{A}_i, 0\le i\le 7$ are all $4\times 4$ matrices, and $\widetilde{C}_8,\widehat{C}_8,\overline{C}_i, 0\le i\le 7$ are all column vectors of length $4$. More specifically, the matrices in \eqref{eq:pceq8} are
{\scriptsize
\begin{align*}
& \overline{A}_0 =
\left[\begin{array}{*{4}{@{\hspace*{0.05in}}c}}
L_0 & L_0-L_1 & \mathbf{0} & \mathbf{0} \\
\mathbf{0} & L_1 & \mathbf{0} & \mathbf{0} \\
\mathbf{0} & \mathbf{0} & L_0 & L_0-L_1 \\
\mathbf{0} & \mathbf{0} & \mathbf{0} & L_1 
\end{array} \hspace*{-0.05in}\right] \ 
\overline{A}_1 =
\left[\begin{array}{*{4}{@{\hspace*{0.05in}}c}}
L_2 & \mathbf{0} & \mathbf{0} & \mathbf{0} \\
L_3-L_2 & L_3 & \mathbf{0} & \mathbf{0} \\
\mathbf{0} & \mathbf{0} & L_2 & \mathbf{0} \\
\mathbf{0} & \mathbf{0} & L_3-L_2 & L_3 
\end{array} \hspace*{-0.05in}\right]
\\
&\overline{A}_2 =
\left[\begin{array}{*{4}{@{\hspace*{0.05in}}c}}
L_4 & \mathbf{0} & \mathbf{0} & \mathbf{0} \\
\mathbf{0} & L_5 & \mathbf{0} & \mathbf{0} \\
\mathbf{0} & \mathbf{0} & L_4 & \mathbf{0} \\
\mathbf{0} & \mathbf{0} & \mathbf{0} & L_5 
\end{array} \hspace*{-0.05in}\right] \ 
\overline{A}_3 =
\left[\begin{array}{*{4}{@{\hspace*{0.05in}}c}}
L_6 & \mathbf{0} & L_6-L_7 & \mathbf{0} \\
\mathbf{0} & L_6 & \mathbf{0} & L_6-L_7 \\
\mathbf{0} & \mathbf{0} & L_7 & \mathbf{0} \\
\mathbf{0} & \mathbf{0} & \mathbf{0} & L_7 
\end{array} \hspace*{-0.05in}\right] 
\\
&\overline{A}_4 =
\left[\begin{array}{*{4}{@{\hspace*{0.05in}}c}}
L_8 & \mathbf{0} & \mathbf{0} & \mathbf{0} \\
\mathbf{0} & L_8 & \mathbf{0} & \mathbf{0} \\
L_9-L_8 & \mathbf{0} & L_9 & \mathbf{0} \\
\mathbf{0} & L_9-L_8 & \mathbf{0} & L_9  
\end{array} \hspace*{-0.05in}\right] \ 
\overline{A}_5 =
\left[\begin{array}{*{4}{@{\hspace*{0.05in}}c}}
L_{10} & \mathbf{0} & \mathbf{0} & \mathbf{0} \\
\mathbf{0} & L_{10} & \mathbf{0} & \mathbf{0} \\
\mathbf{0} & \mathbf{0} & L_{11} & \mathbf{0} \\
\mathbf{0} & \mathbf{0} & \mathbf{0} & L_{11} 
\end{array} \hspace*{-0.05in}\right] \\
& \overline{A}_6 =
\left[\begin{array}{*{4}{@{\hspace*{0.03in}}c}}
L_{12} & \mathbf{0} & \mathbf{0} & \mathbf{0} \\
\mathbf{0} & L_{12} & \mathbf{0} & \mathbf{0} \\
\mathbf{0} & \mathbf{0} & L_{12} & \mathbf{0} \\
\mathbf{0} & \mathbf{0} & \mathbf{0} & L_{12} 
\end{array} \hspace*{-0.05in}\right] \ 
\overline{A}_7 =
\left[\begin{array}{*{4}{@{\hspace*{0.03in}}c}}
L_{15} & \mathbf{0} & \mathbf{0} & \mathbf{0} \\
\mathbf{0} & L_{15} & \mathbf{0} & \mathbf{0} \\
\mathbf{0} & \mathbf{0} & L_{15} & \mathbf{0} \\
\mathbf{0} & \mathbf{0} & \mathbf{0} & L_{15}
\end{array} \hspace*{-0.05in}\right]  
\\
&\widetilde{A}_8 =
\left[\begin{array}{*{4}{@{\hspace*{0.05in}}c}}
L_{16} & \mathbf{0} & \mathbf{0} & \mathbf{0} \\
\mathbf{0} & L_{16} & \mathbf{0} & \mathbf{0} \\
\mathbf{0} & \mathbf{0} & L_{16} & \mathbf{0} \\
\mathbf{0} & \mathbf{0} & \mathbf{0} & L_{16} 
\end{array} \hspace*{-0.05in}\right] \ 
\widehat{A}_8 =
\left[\begin{array}{*{4}{@{\hspace*{0.05in}}c}}
L_{17} & \mathbf{0} & \mathbf{0} & \mathbf{0} \\
\mathbf{0} & L_{17} & \mathbf{0} & \mathbf{0} \\
\mathbf{0} & \mathbf{0} & L_{17} & \mathbf{0} \\
\mathbf{0} & \mathbf{0} & \mathbf{0} & L_{17}
\end{array} \hspace*{-0.05in}\right] 
\end{align*}
}
and the column vectors in \eqref{eq:pceq8} are
\begin{align*}
& \overline{C}_i =
\left[ \begin{array}{c}
  C_i(0)+C_i(4) \\
  C_i(1)+C_i(5) \\
  C_i(2)+C_i(6) \\
  C_i(3)+C_i(7) 
\end{array}
\right]
\text{~for~} 0\le i\le 7 , 
\\
& \widetilde{C}_8 =
\left[ \begin{array}{c}
  C_8(0) \\
  C_8(1) \\
  C_8(2) \\
  C_8(3)
\end{array}
\right] , \quad
\widehat{C}_8 = 
\left[ \begin{array}{c}
  C_8(4) \\
  C_8(5) \\
  C_8(6) \\
  C_8(7)
\end{array}
\right] .
\end{align*}
The matrices $\overline{A}_0,\overline{A}_1,\dots,\overline{A}_5$ are precisely the $6$ parity check matrices that would appear in our MSR code construction with code length $n=6$ and subpacketization $\ell=2^{6/3}=4$. The other $4$ matrices $\widetilde{A}_8,\widehat{A}_8,\overline{A}_6,\overline{A}_7$ are all block-diagonal matrices, and the diagonal entries are the same within each matrix. Moreover, the $\lambda_i$'s (or equivalently $L_i$'s) that appear in $\widetilde{A}_8,\widehat{A}_8,\overline{A}_6,\overline{A}_7$ do not intersect with the $\lambda_i$'s that appear in $\overline{A}_0,\overline{A}_1,\dots,\overline{A}_5$.
The method we used to prove the MDS property of our MSR code construction in Section~\ref{subsect:MDS} can be easily generalized to show that \eqref{eq:pceq8} also defines an MDS array code $(\overline{C}_0,\overline{C}_1,\overline{C}_2,\dots,\overline{C}_7,\widetilde{C}_8,\widehat{C}_8)$ with code length $10$ and code dimension $6$. Therefore, $\widetilde{C}_8$ and $\widehat{C}_8$ can be recovered from any $6$ vectors in the set $\{\overline{C}_0,\overline{C}_1,\dots,\overline{C}_7\}$. Once we know the values of $\widetilde{C}_8$ and $\widehat{C}_8$, we are able to recover all the coordinates of $C_8$.

\section{MDS property} \label{MDS}

In this section, we prove that our code construction allows us to recover any $r=n-k$ node failures. We write the index set of failed nodes as $\mathcal{F}$, whose size is $|\mathcal{F}|=r$. Recall that all the nodes are divided into groups of size $3$. The failed nodes fall into different groups, and we classify these groups into the following seven sets
{\small \begin{align*}
    &\mathcal{G}_1(\mathcal{F})=\{i:0\le i\le n/3-1,3i\in\mathcal{F},3i+1\in\mathcal{F},3i+2\in\mathcal{F}\},\\
    &\mathcal{G}_2(\mathcal{F})=\{i:0\le i\le n/3-1,3i\in\mathcal{F},3i+1\in\mathcal{F},3i+2\not\in \mathcal{F}\},\\
    &\mathcal{G}_3(\mathcal{F})=\{i:0\le i\le n/3-1,3i\in\mathcal{F},3i+1\notin \mathcal{F},3i+2\in\mathcal{F}\},\\
    &\mathcal{G}_4(\mathcal{F})=\{i:0\le i\le n/3-1,3i\notin\mathcal{F},3i+1\in\mathcal{F},3i+2\in\mathcal{F}\},\\
    &\mathcal{G}_5(\mathcal{F})=\{i:0\le i\le n/3-1,3i\in \mathcal{F},3i+1\notin\mathcal{F},3i+2\notin\mathcal{F}\},\\
    &\mathcal{G}_6(\mathcal{F})=\{i:0\le i\le n/3-1,3i\notin\mathcal{F},3i+1\in \mathcal{F},3i+2\notin\mathcal{F}\},\\
    &\mathcal{G}_7(\mathcal{F})=\{i:0\le i\le n/3-1,3i\notin\mathcal{F},3i+1\notin\mathcal{F},3i+2\in \mathcal{F}\}.
\end{align*}}%
By definition, $\mathcal{G}_1(\mathcal{F})$ consists of the groups whose all $3$ nodes fail, $\mathcal{G}_2(\mathcal{F})$ consists of the groups whose first $2$ nodes fail, and so on.
\begin{definition}\label{erasurepattern}
For an integer vector $\vec{z}=(z_1,z_2,z_3,z_4,z_5,z_6,z_7)$, we say that an erasure pattern $\mathcal{F}$ is of type $\vec{z}=(z_1,z_2,z_3,z_4,z_5,z_6,z_7)$ if $|\mathcal{G}_i(\mathcal{F})|=z_i$ for $1\le i \le 7$.
\end{definition}
From the set of all erasure patterns of type $\vec{z}$,
we pick the only erasure pattern $\mathcal{F}(\vec{z})$ that satisfies
\begin{equation} \label{eq:Gi}
{\small\begin{aligned}
& \mathcal{G}_1(\mathcal{F}(\vec{z}))=\{j: 0\le j< z_1 \}, \\
&
\begin{aligned}
	\mathcal{G}_i(\mathcal{F}(\vec{z}))=\{j:z_1+z_2+\dots+z_{i-1}\le j < z_1+z_2+\dots+z_i\} & \\ \text{~for~} 2\le i\le 7,&
\end{aligned} 
\end{aligned}}
\end{equation}
and call it the canonical erasure patterns of type $\vec{z}$.
%and we say that an erasure pattern $\mathcal{F}$ is of type $\vec{z}$ if $|\mathcal{G}_i(\mathcal{F})|=z_i$ for $1\le i \le 7$.
\emph{In Appendix,  we show that if we can recover from the canonical erasure pattern $\mathcal{F}(\vec{z})$, 
then we can recover from all erasure patterns of type $\vec{z}$.}  
It is easy to see that if $\mathcal{F}$ is of type $\vec{z}=(z_1,z_2,z_3,z_4,z_5,z_6,z_7)$, then $|\mathcal{F}|=3z_1+2z_2+2z_3+2z_4+z_5+z_6+z_7$.
Therefore, in order to prove the MDS property, we only need to show that we can recover from the canonical erasure pattern $\mathcal{F}(\vec{z})$ for all $\vec{z}=(z_1,z_2,z_3,z_4,z_5,z_6,z_7)$ satisfying 
\begin{equation} \label{eq:wsmr}
3z_1+2z_2+2z_3+2z_4+z_5+z_6+z_7=r.
\end{equation}

Now let us pick a vector $\vec{z}$ satisfying \eqref{eq:wsmr}. 
To prove that we can recover from $\mathcal{F}(\vec{z})$, we only need to show that if $C_i=\mathbf{0}$ for every $i\notin \mathcal{F}(\vec{z})$, then $C_i=0$ for every $i\in\mathcal{F}(\vec{z})$ is the only solution to the parity check equations \eqref{eq:pc}. When $C_i=\mathbf{0}$ for every $i\notin \mathcal{F}(\vec{z})$, the equations \eqref{eq:pc} can be written in the following matrix form
\begin{equation} \label{eq:pc_redu}
\sum_{i\in\mathcal{F}(\vec{z})} A_i C_i = \mathbf{0} ,
\end{equation}
where each $A_i$ is a $r\ell \times \ell$ matrix defined by \eqref{pcmt}, and $\mathbf{0}$ on the right-hand side is the all-zero column vector of length $r\ell$.

For $0\le i\le n-1$ and $1\le u\le \ell$, we define $A_i^{(u)}$ as the $u r\times u$ submatrix at the top left corner of $A_i$. 
Let us take the matrices $A_0,A_1,\dots,A_8$ defined in Section~\ref{example} as examples: For these matrices, we have
\begin{align*}
&A_0^{(2)}=\begin{bmatrix} L_0 & L_0-L_1 \\ \mathbf{0} & L_1 \end{bmatrix} , \quad
A_1^{(2)}=\begin{bmatrix} L_2 & \mathbf{0} \\ L_3-L_2 & L_3 \end{bmatrix} , \\
&A_3^{(4)}=\begin{bmatrix}
	L_6 & \mathbf{0} & L_6-L_7 & \mathbf{0} \\
	\mathbf{0} & L_6 & \mathbf{0} & L_6-L_7 \\
	\mathbf{0} & \mathbf{0} & L_7 & \mathbf{0} \\
	\mathbf{0} & \mathbf{0} & \mathbf{0} & L_7
\end{bmatrix} ,	
\end{align*}
where $L_i$ is a column vector of length $r$.

According to \eqref{pcmt}, for $0\le i\le 2$, we have 
$$
A_i C_i = \left[ \begin{array}{c}
A_i^{(2)} (C_i(0),C_i(1))^T \\
A_i^{(2)} (C_i(2),C_i(3))^T \\
\vdots \\
A_i^{(2)} (C_i(\ell-2),C_i(\ell-1))^T
\end{array} \right] ,
$$
where $A_i^{(2)}(C_i(2j),C_i(2j+1))^T$ is a column vector of length $2r$. For $0\le i\le 5$, we have
{\small\begin{align*}
	A_i C_i = \left[ \begin{array}{c}
		A_i^{(4)} (C_i(0),C_i(1),C_i(2),C_i(3))^T \\
		A_i^{(4)} (C_i(4),C_i(5),C_i(6),C_i(7))^T \\
		\vdots \\
		A_i^{(4)} (C_i(\ell-4),C_i(\ell-3),C_i(\ell-2),C_i(\ell-1))^T
	\end{array} \right] .
\end{align*}}
In general, for $1\le v\le n/3$ and $0\le i\le 3v-1$, we have
{\footnotesize \begin{equation} \label{eq:a2u}
A_i C_i = \left[ \begin{array}{c}
A_i^{(2^v)} (C_i(0),C_i(1),\dots,C_i(2^v-1))^T \\
A_i^{(2^v)} (C_i(2^v),C_i(2^v+1),\dots,C_i(2^{v+1}-1))^T \\
\vdots \\
A_i^{(2^v)} (C_i(\ell-2^v),C_i(\ell-2^v+1),\dots,C_i(\ell-1))^T
\end{array} \right] .
\end{equation}}

For a vector $\vec{z}=(z_1,z_2,z_3,z_4,z_5,z_6,z_7)$ satisfying \eqref{eq:wsmr}, we write $z=z_1+z_2+\dots+z_7$. 
Let us write the set $\mathcal{F}(\vec{z})$ as $\mathcal{F}(\vec{z})=\{i_1,i_2,\dots,i_r\}$, where $i_1<i_2<\dots<i_r$, and we define a $2^z r\times 2^z r$ matrix $M_{\vec{z}}$ as follows: 
\begin{equation}  \label{eq:mzv}
M_{\vec{z}} = \left[ A_{i_1}^{(2^z)} \quad A_{i_2}^{(2^z)} \quad \dots \quad A_{i_r}^{(2^z)} \right] .
\end{equation}
\begin{lemma}
For a vector $\vec{z}$ satisfying \eqref{eq:wsmr}, the equation \eqref{eq:pc_redu} only has zero solution if and only if the matrix $M_{\vec{z}}$ is invertible.
\end{lemma}
\begin{proof}
By definition \eqref{eq:Gi}, we have $i\le 3z-1$ for all $i\in \mathcal{F}(\vec{z})$. Therefore, we can use \eqref{eq:a2u} to decompose \eqref{eq:pc_redu} into the following $\ell/2^z$ equations:
{\footnotesize \begin{align*}
\sum_{i\in\mathcal{F}(\vec{z})} A_i^{(2^z)} (C_i(2^z \cdot j),C_i(2^z \cdot j+1),\dots,C_i(2^{z} (j+1) -1))^T = \mathbf{0}&
\\
\text{for~} 0\le j\le \ell/2^z-1 .&
\end{align*}}%
It is clear that these $\ell/2^z$ equations only have zero solution if and only if $M_{\vec{z}}$ is invertible.
\end{proof}

Next we prove that the matrix $M_{\vec{z}}$ is invertible for all $\vec{z}$ satisfying \eqref{eq:wsmr}, and we divide the proof into seven cases.

\textbf{Case 1: $\vec{z}=(z_1,0,0,0,0,0,0)$.}
In this case, \eqref{eq:wsmr} implies that $z_1=r/3$ and $\mathcal{F}(\vec{z})=\{0,1,2,\dots,r-1\}$, and the size of $M_{\vec{z}}$ is $2^{z_1} r\times 2^{z_1} r$.
In order to prove that $M_{\vec{z}}$ is invertible, let us consider a column vector 
\begin{equation} \label{eq:ydb}
\begin{aligned}
	\vec{y}=(&y_{0,0},y_{0,1}\cdots ,y_{0,2^{z_1}-1},y_{1,0},y_{1,1}\cdots ,y_{1,2^{z_1}-1},\cdots,\\
	&y_{r-1,0},y_{r-1,1},\cdots,y_{r-1,2^{z_1}-1})^T
\end{aligned}
\end{equation}
of length $2^{z_1} r$. Next we assume $M_{\vec{z}}~\vec{y}=\mathbf{0}$, and we will prove that $\vec{y}=\mathbf{0}$ is the only solution.
The proof is divided into two steps: The first step is to prove that the following set of coordinates
\begin{equation} \label{eq:st1}
	\begin{aligned}
		&\{y_{3i,b}: 0\le i\le z_1-1, 0\le b\le 2^{z_1}-1, b_i=1\}\\
		\cup &\{y_{3i+1,b}: 0\le i\le z_1-1, 0\le b\le 2^{z_1}-1, b_i=0\}
	\end{aligned}
\end{equation}
are all zero, and the second step is to prove that the remaining coordinates are all zero. The condition $b_i=1$ (or $0$) above means that the $i$-th digit in the binary expansion of $b$ is $1$ (or $0$).

For the first step, we define a sequence of polynomials
\begin{equation} \label{eq:gai}
    g_{a,i}(x)=\begin{cases}
        (x-\lambda_{6i})(x-\lambda_{6i+4}), \text{~if~} a_i= 0 , \\
        (x-\lambda_{6i+3})(x-\lambda_{6i+5}), \text{~if~} a_i=1, \\
    \end{cases}
\end{equation}
for $0\le i\le z_1-1$ and $0\le a\le 2^{z_1}-1$. We define another sequence of polynomials
\begin{equation} \label{eq:defa}
f_a(x)=\prod_{i=0}^{z_1-1}g_{a,i}(x)
\end{equation}
for $0\le a\le 2^{z_1}-1$. Since each $g_{a,i}$ is a quadratic polynomial, the degree of $f_a$ is $2 z_1$, and we write its coefficients as $f_{a,0},f_{a,1},\dots,f_{a,2z_1}$, i.e., we write $f_a(x)=f_{a,0}+f_{a,1}x+\cdots+f_{a,2z_1}x^{2z_1}=\sum_{i=0}^{2z_1}f_{a,i}x^i$.

For each $0\le a\le 2^{z_1}-1$, we define an $(r-2z_1)\times r=z_1\times r$ matrix $F_a$ as
{\footnotesize$$
\left[\begin{array}{*{12}{@{\hspace*{0.05in}}c}}
f_{a,0}&f_{a,1}&\cdots&\cdots&\cdots&\cdots&\cdots&f_{a,2z_1}&0&0&\cdots&0\\
    0&f_{a,0}&f_{a,1}&\cdots&\cdots&\cdots&\cdots&\cdots&f_{a,2z_1}&0&\cdots&0\\
    \vdots&\vdots&\vdots&\vdots&\vdots&\vdots&\vdots&\vdots&\vdots&\vdots&\vdots&\vdots \\
    0&0&\cdots&0&f_{a,0}&f_{a,1}&\cdots&\cdots&\cdots&\cdots&\cdots&f_{a,2z_1}
\end{array}\hspace*{-0.05in}\right].
$$}

Next we define a $2^{z_1} z_1\times 2^{z_1} r$ matrix $F$ as
$$
F := \begin{bmatrix}
    F_0 & 0 & 0 & \cdots & 0 \\
    0 & F_1 & 0 & \cdots & 0 \\
    0 & 0 & F_2 & \cdots & 0 \\
    \vdots & \vdots & \vdots & \vdots & \vdots \\
    0 & 0 & 0 & \cdots & F_{2^{z_1}-1} 
\end{bmatrix} ,
$$
where each $0$ is a $z_1\times r$ matrix with all zero entries.

\begin{figure*}[htbp]
	\begin{equation}\label{eq:add1}
		{\footnotesize \begin{aligned}
				M_{\vec{z}}=
				\left[\begin{array}{c @{\hspace{0.8ex}} c @{\hspace{0ex}} c@{\hspace{0.8ex}} c@{\hspace{0ex}} c@{\hspace{0.8ex}} c@{\hspace{0.8ex}} c@{\hspace{0.8ex}} c@{\hspace{0ex}} c@{\hspace{0ex}} c@{\hspace{0ex}} c@{\hspace{0ex}} c@{\hspace{0ex}} c@{\hspace{0ex}} c@{\hspace{0.4ex}} c@{\hspace{0ex}} c@{\hspace{0.2ex}} c@{\hspace{0.2ex}} c@{\hspace{0.2ex}} c@{\hspace{0ex}}c@{\hspace{0ex}} c@{\hspace{0ex}} c@{\hspace{0ex}} c@{\hspace{0ex}}c}
					L_0&L_0-L_1&\mathbf{0}&\mathbf{0} &L_2&\mathbf{0}&\mathbf{0}&\mathbf{0} &L_4&\mathbf{0}&\mathbf{0}&\mathbf{0} &L_6&\mathbf{0}&L_6-L_7&\mathbf{0} &L_8&\mathbf{0}&\mathbf{0}&\mathbf{0} & L_{10} & \mathbf{0} & \mathbf{0} & \mathbf{0}\\
					\mathbf{0}&L_1&\mathbf{0}&\mathbf{0} &L_3-L_2&L_3&\mathbf{0}&\mathbf{0} &\mathbf{0}&L_5&\mathbf{0}&\mathbf{0} &\mathbf{0}&L_6&\mathbf{0}&L_6-L_7 &\mathbf{0}&L_8&\mathbf{0}&\mathbf{0} & \mathbf{0} & L_{10} & \mathbf{0} & \mathbf{0}\\
					\mathbf{0}&\mathbf{0}&L_0&L_0-L_1 &\mathbf{0}&\mathbf{0}&L_2&\mathbf{0} &\mathbf{0}&\mathbf{0}&L_4&\mathbf{0} &\mathbf{0}&\mathbf{0}&L_7&\mathbf{0} &L_9-L_8&\mathbf{0}&L_9&\mathbf{0} & \mathbf{0} & \mathbf{0} & L_{11} & \mathbf{0}\\
					\mathbf{0}&\mathbf{0}&\mathbf{0}&L_1 &\mathbf{0}&\mathbf{0}&L_3-L_2&L_3 &\mathbf{0}&\mathbf{0}&\mathbf{0}&L_5 &\mathbf{0}&\mathbf{0}&\mathbf{0}&L_7 &\mathbf{0}&L_9-L_8&\mathbf{0}&L_9 & \mathbf{0} & \mathbf{0} & \mathbf{0} & L_{11}
				\end{array}\right] 
		\end{aligned}}
	\end{equation}
\end{figure*}

\begin{figure*}[htbp]
	\begin{equation}\label{eq:add2}
		{\footnotesize \begin{aligned}
				F M_{\vec{z}}= 
				\left[\begin{array}{c@{\hspace{0.8ex}}c@{\hspace{0.8ex}}c@{\hspace{0.8ex}}c@{\hspace{0.8ex}}c@{\hspace{0.8ex}}c@{\hspace{0.8ex}}c@{\hspace{0.8ex}}c@{\hspace{0ex}}c@{\hspace{0ex}}c@{\hspace{0ex}}c@{\hspace{0ex}}c@{\hspace{0ex}}c@{\hspace{0ex}}c@{\hspace{0.8ex}}c@{\hspace{0.8ex}}c@{\hspace{0.8ex}}c@{\hspace{0.8ex}}c@{\hspace{0.8ex}}c@{\hspace{0ex}}c@{\hspace{0ex}}c@{\hspace{0ex}}c@{\hspace{0ex}}c@{\hspace{0ex}}c}
					\mathbf{0}&-f_0(\lambda_1)L_1'&\mathbf{0}&\mathbf{0} &f_0(\lambda_2)L_2'&\mathbf{0}&\mathbf{0}&\mathbf{0} &\mathbf{0}&\mathbf{0}&\mathbf{0}&\mathbf{0} &\mathbf{0}&\mathbf{0}&-f_0(\lambda_7)L_7'&\mathbf{0} &f_0(\lambda_8)L_8'&\mathbf{0}&\mathbf{0}&\mathbf{0} & \mathbf{0} & \mathbf{0} & \mathbf{0} & \mathbf{0}\\
					\mathbf{0}&f_1(\lambda_1)L_1'&\mathbf{0}&\mathbf{0} &-f_1(\lambda_2)L_2'&\mathbf{0}&\mathbf{0}&\mathbf{0} &\mathbf{0}&\mathbf{0}&\mathbf{0}&\mathbf{0} &\mathbf{0}&\mathbf{0}&\mathbf{0}&-f_1(\lambda_7)L_7' &\mathbf{0}&f_1(\lambda_8)L_8'&\mathbf{0}&\mathbf{0} & \mathbf{0} & \mathbf{0} & \mathbf{0} & \mathbf{0}\\
					\mathbf{0}&\mathbf{0}&\mathbf{0}&-f_2(\lambda_1)L_1' &\mathbf{0}&\mathbf{0}&f_2(\lambda_2)L_2'&\mathbf{0} &\mathbf{0}&\mathbf{0}&\mathbf{0}&\mathbf{0} &\mathbf{0}&\mathbf{0}&f_2(\lambda_7)L_7'&\mathbf{0} &-f_2(\lambda_8)L_8'&\mathbf{0}&\mathbf{0}&\mathbf{0} & \mathbf{0} & \mathbf{0} & \mathbf{0} & \mathbf{0}\\
					\mathbf{0}&\mathbf{0}&\mathbf{0}&f_3(\lambda_1)L_1' &\mathbf{0}&\mathbf{0}&-f_3(\lambda_2)L_2'&\mathbf{0} &\mathbf{0}&\mathbf{0}&\mathbf{0}&\mathbf{0} &\mathbf{0}&\mathbf{0}&\mathbf{0}&f_3(\lambda_7)L_7' &\mathbf{0}&-f_3(\lambda_8)L_8'&\mathbf{0}&\mathbf{0} & \mathbf{0} & \mathbf{0} & \mathbf{0} & \mathbf{0}
				\end{array}\right] 
		\end{aligned}}
	\end{equation}
\end{figure*}

\begin{figure*}[htbp]
	\begin{equation}\label{eq:add3}
		{\footnotesize
			\begin{aligned}
				Q=
				\left[\begin{array}{cccccccc}
					-f_0(\lambda_1)L_1' & \mathbf{0} & f_0(\lambda_2)L'_2 & \mathbf{0} & -f_0(\lambda_7)L'_7 & \mathbf{0} & f_0(\lambda_8)L'_8 & \mathbf{0} \\
					f_1(\lambda_1)L'_1 & \mathbf{0} & -f_1(\lambda_2)L'_2 & \mathbf{0} & \mathbf{0} & -f_1(\lambda_7)L'_7 & \mathbf{0} & f_1(\lambda_8)L'_8 \\
					\mathbf{0} & -f_2(\lambda_1)L'_1 & \mathbf{0} & f_2(\lambda_2)L'_2 & f_2(\lambda_7)L'_7 & \mathbf{0} & -f_2(\lambda_8)L'_8 & \mathbf{0} \\
					\mathbf{0} & f_3(\lambda_1)L'_1 & \mathbf{0} & -f_3(\lambda_2)L'_2 & \mathbf{0} & f_3(\lambda_7)L'_7 & \mathbf{0} & -f_3(\lambda_8)L'_8 
				\end{array}\right] 
			\end{aligned}
		}
	\end{equation}
\end{figure*}

The equation $M_{\vec{z}}~\vec{y}=\mathbf{0}$ implies that $F M_{\vec{z}}~\vec{y}=\mathbf{0}$. The product $F M_{\vec{z}}$ is a $2^{z_1} z_1\times 2^{z_1} r$ matrix, and we represent it in the following form
\begin{equation} \label{eq:fmz}
F M_{\vec{z}} = \begin{bmatrix}
        B_0 & B_1 & \cdots & B_{r-1}
    \end{bmatrix} ,
\end{equation}
where each $B_i$ is a $2^{z_1} z_1\times 2^{z_1}$ matrix. We further represent each matrix $B_i$ as
{\small $$
\begin{bmatrix}
    B_i(0,0) & B_i(0,1) & \cdots & B_i(0,2^{z_1}-1) \\
    B_i(1,0) & B_i(1,1) & \cdots & B_i(1,2^{z_1}-1) \\
    \vdots & \vdots & \vdots & \vdots \\
    B_i(2^{z_1}-1,0) & B_i(2^{z_1}-1,1) & \cdots & B_i(2^{z_1}-1,2^{z_1}-1) 
\end{bmatrix} ,
$$}%
where each $B_i(a,b)$ is a column vector of length $z_1$ for $0\le i\le r-1$ and $0\le a,b\le 2^{z_1}-1$. More precisely, for every $0\le i\le z_1-1$ and $0\le a,b\le 2^{z_1}-1$, we have
\begin{equation} \label{eq:sks}
{\small \begin{aligned}
B_{3i}(a,b) &=\left\{
\begin{array}{ll}
  f_a(\lambda_{6i+1})L'_{6i+1}   & \mbox{if~} a=b~\text{and}~a_i=b_i=1, \\
  -f_a(\lambda_{6i+1})L'_{6i+1}   & \mbox{if~} a_i=0, b_i=1 \\ &\mbox{and~} a_j=b_j \,\forall j\neq i, \\
  \mathbf{0} & \mbox{otherwise},
\end{array}
\right.  \\
B_{3i+1}(a,b) &=\left\{
\begin{array}{ll}
  f_a(\lambda_{6i+2})L'_{6i+2}   & \mbox{if~} a=b~\text{and}~a_i=b_i=0, \\
  -f_a(\lambda_{6i+2})L'_{6i+2}   & \mbox{if~} a_i=1, b_i=0 \\ &\mbox{and~} a_j=b_j \,\forall j\neq i,  \\
  \mathbf{0} & \mbox{otherwise},
\end{array}
\right. \\
B_{3i+2}(a,b) &= \mathbf{0} \qquad \quad \mbox{for all~} a \mbox{~and~} b ,
\end{aligned}}
\end{equation}
where the length of each all-zero vector is $z_1$, and 
\begin{equation} \label{eq:cvn}
L'_{6i+1} = \begin{bmatrix}
1 \\
\lambda_{6i+1} \\
\lambda_{6i+1}^2 \\
\vdots \\
\lambda_{6i+1}^{z_1-1}
\end{bmatrix},~
L'_{6i+2} = \begin{bmatrix}
1 \\
\lambda_{6i+2} \\
\lambda_{6i+2}^2 \\
\vdots \\
\lambda_{6i+2}^{z_1-1}
\end{bmatrix}
\text{~~for~} 0\le i\le z_1-1.
\end{equation}
From \eqref{eq:sks} we can see that $B_{3i+2}$ is an all-zero matrix for every $0\le i\le z_1-1$. Moreover, $B_{3i}$ and $B_{3i+1}$ have exactly half of their columns to be nonzero: The $b$-th column of $B_{3i}$ is nonzero if and only if $b_i=1$; the $b$-th column of $B_{3i+1}$ is nonzero if and only if $b_i=0$. Therefore, the matrix $F M_{\vec{z}}$ in \eqref{eq:fmz} has  $2^{z_1} z_1$ nonzero columns, and these nonzero columns are multiplied with the coordinates in \eqref{eq:st1}.
We use $Q$ to denote the $2^{z_1} z_1 \times 2^{z_1} z_1$ matrix formed by the $2^{z_1} z_1$ nonzero columns of $F M_{\vec{z}}$, and we use $\vec{y}^{(1)}$ to denote the subvector of $\vec{y}$ formed by the $2^{z_1} z_1$ coordinates in \eqref{eq:st1}. Then the equation $F M_{\vec{z}}~\vec{y}=\mathbf{0}$ is equivalent to $Q \vec{y}^{(1)}=\mathbf{0}$. In Lemma~\ref{lm:det}, we prove that $Q$ is invertible, which immediately implies that the coordinates in \eqref{eq:st1} are all zero. Before presenting Lemma~\ref{lm:det}, we first give an example of the matrices $M_{\vec{z}},FM_{\vec{z}}$ and $Q$.

\begin{example} \label{ex:1}
Let us take $\vec{z}=(2,0,0,0,0,0,0)$. Then $r=3z_1=6$. In this case, $M_{\vec{z}}$ is matrix \eqref{eq:add1}, 
where each $L_i$ is a column vector of length $r=6$, and each $\mathbf{0}$ is an all-zero column vector of length $r=6$. After multiplying $F$ with $M_{\vec{z}}$, we obtain equation \eqref{eq:add2}, 
where each $L_i'$ is a column vector of length $2$, and each $\mathbf{0}$ is an all-zero column vector of length $2$. Therefore, $Q$ is matrix \eqref{eq:add3}. 
\end{example}

\begin{lemma} \label{lm:det}
Let $Q$ be the $2^{z_1} z_1 \times 2^{z_1} z_1$ matrix formed by the $2^{z_1} z_1$ nonzero columns of $F M_{\vec{z}}$. Then $Q$ is invertible. 
\end{lemma}
\begin{proof}
We write $Q$ in the form 
$$
Q = \begin{bmatrix}
        D_0 & D_1 & D_2 & \cdots & D_{2z_1-1}
    \end{bmatrix} ,
$$
where each $D_i$ is a $2^{z_1} z_1\times 2^{z_1-1}$ matrix. We further write each matrix $D_i$ as
\begin{equation}  \label{eq:bsl}
{\footnotesize \begin{bmatrix}
    D_i(0,0) & D_i(0,1) & \cdots & D_i(0,2^{z_1-1}-1) \\
    D_i(1,0) & D_i(1,1) & \cdots & D_i(1,2^{z_1-1}-1) \\
    \vdots & \vdots & \vdots & \vdots \\
    D_i(2^{z_1}-1,0) & D_i(2^{z_1}-1,1) & \cdots & D_i(2^{z_1}-1,2^{z_1-1}-1) 
\end{bmatrix} ,}
\end{equation}
where each $D_i(a,b)$ is a column vector of length $z_1$ for $0\le i\le 2z_1-1, 0\le a\le 2^{z_1}-1$ and $0\le b\le 2^{z_1-1}-1$. More precisely, for every $0\le i\le z_1-1, 0\le a\le 2^{z_1}-1$ and every $0\le b\le 2^{z_1-1}-1$, we have equations \eqref{eq:add5}, 
\begin{figure*}
	\begin{equation}\label{eq:add5}
		\begin{aligned}
			D_{2i}(a,b)=\left\{\begin{array}{ll}
				-f_a(\lambda_{6i+1})L'_{6i+1} & \mbox{if~} a=(b_{z_1-2},b_{z_1-3},\cdots ,b_i,0,b_{i-1},b_{i-2},\cdots,b_0), \\
				f_a(\lambda_{6i+1})L'_{6i+1} & \mbox{if~} a=(b_{z_1-2},b_{z_1-3},\cdots ,b_i,1,b_{i-1},b_{i-2},\cdots,b_0), \\
				\mathbf{0} & \mbox{otherwise},
			\end{array}\right.  \\
			D_{2i+1}(a,b)=\left\{\begin{array}{ll}
				f_a(\lambda_{6i+2})L'_{6i+2} & \mbox{if~} a=(b_{z_1-2},b_{z_1-3},\cdots ,b_i,0,b_{i-1},b_{i-2},\cdots,b_0),\\
				-f_a(\lambda_{6i+2})L'_{6i+2} & \mbox{if~} a=(b_{z_1-2},b_{z_1-3},\cdots ,b_i,1,b_{i-1},b_{i-2},\cdots,b_0), \\
				\mathbf{0} & \mbox{otherwise},
			\end{array}\right.
		\end{aligned}
	\end{equation}
\end{figure*}
where 
%$0$ in the last line denotes the all-zero column vector of length $z_1$; 
the vectors $L'_{6i+1}$ and $L'_{6i+2}$ are defined in \eqref{eq:cvn}. Note that the binary expansion of $a$ has one more digit than that of $b$ because the range of $a$ is larger than the range of $b$.

For each $0\le i\le z_1-1$ and $0\le b\le 2^{z_1-1}-1$, we introduce the shorthand notation
\begin{align*}
a(i,0,b):=(b_{z_1-2},b_{z_1-3},\cdots ,b_i,0,b_{i-1},b_{i-2},\cdots,b_0) , \\
a(i,1,b):=(b_{z_1-2},b_{z_1-3},\cdots ,b_i,1,b_{i-1},b_{i-2},\cdots,b_0) .
\end{align*}
Two corner cases are worth mentioning: When $i=0$ or $z_1-1$, we have
\begin{align*}
a(0,0,b)&=(b_{z_1-2},b_{z_1-3},\cdots,b_0,0) , \\
a(0,1,b)&=(b_{z_1-2},b_{z_1-3},\cdots,b_0,1) ,  \\
a(z_1-1,0,b)&=(0,b_{z_1-2},b_{z_1-3},\cdots,b_0) , \\
a(z_1-1,1,b)&=(1,b_{z_1-2},b_{z_1-3},\cdots,b_0) .
\end{align*}
Next we will construct a $2^{z_1} z_1 \times 2^{z_1} z_1$ matrix 
\begin{equation}  \label{eq:qbr}
\overline{Q} = \begin{bmatrix}
        \overline{D}_0 & \overline{D}_1 & \overline{D}_2 & \cdots & \overline{D}_{2z_1-1}
    \end{bmatrix} ,
\end{equation}
where each $\overline{D}_i$ has the same size as $D_i$. For $0\le i\le 2z_1-1$, the matrix $\overline{D}_i$ is obtained from multiplying a nonzero element to each column of $D_i$. For $0\le i\le 2z_1-1, 0\le a\le 2^{z_1}-1$ and $0\le b\le 2^{z_1-1}-1$, we define the column vector $\overline{D}_i(a,b)$ in the same way as $D_i(a,b)$ in \eqref{eq:bsl}.

With the above notation at hand, we are ready to explain how to construct $\overline{D}_i$, or equivalently, how to construct $\overline{Q}$. For $0\le i\le z_1-1$ and $0\le b\le 2^{z_1-1}-1$, the $b$th column of $\overline{D}_{2i}$ is obtained from multiplying $-1/f_{a(i,0,b)}(\lambda_{6i+1})$ to the $b$th column of $D_{2i}$, and the $b$th column of $\overline{D}_{2i+1}$ is obtained from multiplying $1/f_{a(i,0,b)}(\lambda_{6i+2})$ to the $b$th column of $D_{2i+1}$. Therefore, 
\begin{equation} \label{eq:dbar1}
\begin{aligned}
	\overline{D}_{2i}(a(i,0,b),b) &= L'_{6i+1}, 
	\\
	\overline{D}_{2i}(a(i,1,b),b) &= -\frac{f_{a(i,1,b)}(\lambda_{6i+1})}{f_{a(i,0,b)}(\lambda_{6i+1})} L'_{6i+1} ,
\end{aligned}
\end{equation}
and $\overline{D}_{2i}(a,b)$ is an all-zero vector for all other $a\neq a(i,0,b),a(i,1,b)$. Similarly, 
\begin{equation} \label{eq:dbar2}
\begin{aligned}
	\overline{D}_{2i+1}(a(i,0,b),b) &= L'_{6i+2}, 
	\\
	\overline{D}_{2i+1}(a(i,1,b),b) &= -\frac{f_{a(i,1,b)}(\lambda_{6i+2})}{f_{a(i,0,b)}(\lambda_{6i+2})} L'_{6i+2} ,
\end{aligned}
\end{equation}
and $\overline{D}_{2i+1}(a,b)$ is an all-zero vector for all other $a\neq a(i,0,b),a(i,1,b)$.

By definition \eqref{eq:gai}, the polynomial $g_{a,i}$ only depends on the $i$th digit in the binary expansion of $a$. Since $a(i,0,b)$ and $a(i,1,b)$ only differ in the $i$th digit in their binary expansions, we have $g_{a(i,0,b),j}(x)=g_{a(i,1,b),j}(x)$ for all $j\neq i$. Then by \eqref{eq:defa}, we have
\begin{equation} \label{eq:r1}
	\begin{aligned}
		&-\frac{f_{a(i,1,b)}(\lambda_{6i+1})}{f_{a(i,0,b)}(\lambda_{6i+1})}\\
		= &-\frac{g_{a(i,1,b),i}(\lambda_{6i+1})}{g_{a(i,0,b),i}(\lambda_{6i+1})} \\
		=
		&-\frac{(\lambda_{6i+1}-\lambda_{6i+3})(\lambda_{6i+1}-\lambda_{6i+5})}{(\lambda_{6i+1}-\lambda_{6i})(\lambda_{6i+1}-\lambda_{6i+4})}\\
		=
		&\gamma_{6i+1} ,
	\end{aligned}
\end{equation}
where the last equality follows from the definition \eqref{eq:gama}. 
Similarly, 
\begin{equation} \label{eq:r2}
\begin{aligned}
	&-\frac{f_{a(i,1,b)}(\lambda_{6i+2})}{f_{a(i,0,b)}(\lambda_{6i+2})}\\
	= 
	&-\frac{g_{a(i,1,b),i}(\lambda_{6i+2})}{g_{a(i,0,b),i}(\lambda_{6i+2})} \\
	= &-\frac{(\lambda_{6i+2}-\lambda_{6i+3})(\lambda_{6i+2}-\lambda_{6i+5})}{(\lambda_{6i+2}-\lambda_{6i})(\lambda_{6i+2}-\lambda_{6i+4})}\\
	=
	&\gamma_{6i+2}
\end{aligned}
\end{equation}
where the last equality follows from the definition \eqref{eq:gama}. Taking \eqref{eq:r1}--\eqref{eq:r2} into \eqref{eq:dbar1}--\eqref{eq:dbar2}, we obtain 
\begin{equation}  \label{eq:fdb}
   \begin{aligned}
   		\overline{D}_{2i}(a,b)=&\left\{\begin{array}{ll}
   			L'_{6i+1} & \mbox{if~} a=a(i,0,b),\\
   			\gamma_{6i+1} L'_{6i+1} & \mbox{if~} a=a(i,1,b), \\
   			\mathbf{0} & \mbox{otherwise},
   		\end{array}\right.  
   		\\
   		\overline{D}_{2i+1}(a,b)=&\left\{\begin{array}{ll}
   			L'_{6i+2} & \mbox{if~} a=a(i,0,b), \\
   			\gamma_{6i+2} L'_{6i+2} & \mbox{if~} a=a(i,1,b),\\
   			\mathbf{0} & \mbox{otherwise},
   		\end{array}\right.
   \end{aligned}
\end{equation}
for $0\le i\le z_1-1, 0\le a\le 2^{z_1}-1$ and $0\le b\le 2^{z_1-1}-1$.

Since all the $\lambda_i$'s in our code construction are distinct, \eqref{eq:gai}--\eqref{eq:defa} imply that $-1/f_{a(i,0,b)}(\lambda_{6i+1})$ and $1/f_{a(i,0,b)}(\lambda_{6i+2})$, i.e., the elements multiplied to each column of $Q$, are nonzero. Therefore, $Q$ is invertible if and only if $\overline{Q}$ is invertible. 

Before proceeding to prove that $\overline{Q}$ is invertible, we first give a concrete example of $\overline{Q}$: For the choice of parameters in Example~\ref{ex:1}, we have
{\scriptsize \begin{align*}
    \overline{Q}=
    \left[\begin{array}{cccccccc}
    L_1' & \mathbf{0} & L'_2 & \mathbf{0} & L'_7 & \mathbf{0} & L'_8 & \mathbf{0} \\
    \gamma_1 L'_1 & \mathbf{0} & \gamma_2 L'_2 & \mathbf{0} & \mathbf{0} & L'_7 & \mathbf{0} & L'_8 \\
    \mathbf{0} & L'_1 & \mathbf{0} & L'_2 & \gamma_7 L'_7 & \mathbf{0} & \gamma_8 L'_8 & \mathbf{0} \\
    \mathbf{0} & \gamma_1 L'_1 & \mathbf{0} & \gamma_2 L'_2 & \mathbf{0} & \gamma_7 L'_7 & \mathbf{0} & \gamma_8 L'_8 
\end{array}\right] .
\end{align*}}

To prove that $\overline{Q}$ is invertible, we only need to show that $\det(\overline{Q})\neq 0$. From \eqref{eq:fdb} we can see that $\det(\overline{Q})$ is a polynomial of $\lambda_{6i+1},\lambda_{6i+2},\gamma_{6i+1},\gamma_{6i+2}$ for $0\le i \le z_1-1$. The variables $\lambda_{6i+1}$ and $\gamma_{6i+1}$ appear in every column of $\overline{D}_{2i}$, and they do not appear in any other $\overline{D}_j$ for $j\neq 2i$. The maximum degree of $\lambda_{6i+1}$ in each column of $\overline{D}_{2i}$ is $z_1-1$, and the maximum degree of $\gamma_{6i+1}$ in each column of $\overline{D}_{2i}$ is $1$. The matrix $\overline{D}_{2i}$ has $2^{z_1-1}$ columns. Therefore, the degree of $\lambda_{6i+1}$ in $\det(\overline{Q})$ is at most $(z_1-1) 2^{z_1-1}$, and the degree of $\gamma_{6i+1}$ in $\det(\overline{Q})$ is at most $2^{z_1-1}$.
Using similar arguments we can show that the degree of $\lambda_{6i+2}$ in $\det(\overline{Q})$ is at most $(z_1-1) 2^{z_1-1}$, and the degree of $\gamma_{6i+2}$ in $\det(\overline{Q})$ is at most $2^{z_1-1}$.

Next we prove that $(\lambda_{6i+1}-\lambda_{6j+1})^{2^{z_1-2}}$ is a factor of $\det(\overline{Q})$ for every pair of $i\neq j$. For $0\le i\le 2z_1-1$ and $0\le b\le 2^{z_1-1}-1$, we use $\overline{D}_i(b)$ to denote the $b$th column of the matrix $\overline{D}_i$. For $0\le i<j\le z_1-1$ and $0\le b\le 2^{z_1-1}-1$, we define
{\small \begin{align*}
	&\phi(i,j,b)\\
	=&(b_{z_1-2},b_{z_1-3},\dots,b_j,b_{j-2},b_{j-3},\dots,b_i,0,b_{i-1},b_{i-2},\dots,b_0).
\end{align*}}
Three corner cases are worth mentioning: (i) If $j=z_1-1$, then $b_{z_1-2},b_{z_1-3},\dots,b_j$ is an empty subvector, i.e., we remove this part from the above definition; (ii) if $i=j-1$, then $b_{j-2},b_{j-3},\dots,b_i$ is an empty subvector; (iii) if $i=0$, then $b_{i-1},b_{i-2},\dots,b_0$ is an empty subvector. 
For $0\le i\le z_1-2$, we define a set
$$
\mathcal{T}_i:=\{b:0\le b\le 2^{z_1-1}-1, b_i=0\}.
$$

Given a pair of $i,j$ such that $i<j$, we define a $2^{z_1} z_1\times 2^{z_1-1}$ matrix $\overline{D}'_{2i}$ as follows: For $0\le b\le 2^{z_1-1}-1$, we write the $b$th column of $\overline{D}'_{2i}$ as $\overline{D}'_{2i}(b)$. If $b\in\mathcal{T}_{j-1}$, then
\begin{equation} \label{eq:dmd}
\begin{aligned}
	\overline{D}'_{2i}(b)=&\overline{D}_{2i}(b) + \gamma_{6j+1} \overline{D}_{2i}(b+2^{j-1}) \\-& \overline{D}_{2j}(\phi(i,j,b)) - \gamma_{6i+1} \overline{D}_{2j}(\phi(i,j,b)+2^i) ;
\end{aligned}
\end{equation}

%\begin{align*}
%    b&=(b_{z_1-2},b_{z_1-3},\cdots,b_{j},0,b_{j-2},\cdots,b_0),\\
%    b+2^{j-1}&=(b_{z_1-2},b_{z_1-3},\cdots,b_{j},1,b_{j-2},\cdots,b_0),\\
%    \phi(i,j,b)&=(b_{z_1-2},b_{z_1-3},\dots,b_j,b_{j-2},b_{j-3},\dots,b_i,0,b_{i-1},b_{i-2},\dots,b_0),\\
%    \phi(i,j,b)+2^i&=(b_{z_1-2},b_{z_1-3},\dots,b_j,b_{j-2},b_{j-3},\dots,b_i,1,b_{i-1},b_{i-2},\dots,b_0).
%\end{align*}

If $b\notin\mathcal{T}_{j-1}$, then $\overline{D}'_{2i}(b)=\overline{D}_{2i}(b)$. We further define a $2^{z_1} z_1 \times 2^{z_1} z_1$ matrix $\overline{Q}'$ obtained from replacing $\overline{D}_{2i}$ with $\overline{D}'_{2i}$ on the right-hand side of \eqref{eq:qbr}. It is easy to see that $\det(\overline{Q}')=\det(\overline{Q})$. Next we show that for every $b\in\mathcal{T}_{j-1}$, $(\lambda_{6i+1}-\lambda_{6j+1})$ is a common factor of all the entries in the column $\overline{D}'_{2i}(b)$. To see this, we write $\overline{D}'_{2i}(b)$ as
$$
\overline{D}'_{2i}(b) = \begin{bmatrix}
\overline{D}'_{2i}(0,b) \\
\overline{D}'_{2i}(1,b) \\
\vdots \\
\overline{D}'_{2i}(2^{z_1}-1,b)
\end{bmatrix} ,
$$
where each $\overline{D}'_{2i}(a,b)$ is a column vector of length $z_1$ for $0\le a\le 2^{z_1}-1$. 
Note that for $b\in \mathcal{T}_{j-1}$ we have
{\small \begin{align*}
    &b=\\&(b_{z_1-2},b_{z_1-3},\dots,b_j,0,b_{j-2},b_{j-3},\dots,b_i,b_{i-1},b_{i-2},\dots,b_0),\\
    &b+2^{j-1}=\\&(b_{z_1-2},b_{z_1-3},\dots,b_j,1,b_{j-2},b_{j-3},\dots,b_i,b_{i-1},b_{i-2},\dots,b_0),\\
    &\phi(i,j,b)=\\&(b_{z_1-2},b_{z_1-3},\dots,b_j,b_{j-2},b_{j-3},\dots,b_i,0,b_{i-1},b_{i-2},\dots,b_0),\\
    &\phi(i,j,b)+2^i=\\&(b_{z_1-2},b_{z_1-3},\dots,b_j,b_{j-2},b_{j-3},\dots,b_i,1,b_{i-1},b_{i-2},\dots,b_0),\\
    &a(i,0,b)=\\&(b_{z_1-2},b_{z_1-3},\dots,b_j,0,b_{j-2},b_{j-3},\dots,b_i,0,b_{i-1},b_{i-2},\dots,b_0),\\
    &a(i,0,b)+2^i=\\&(b_{z_1-2},b_{z_1-3},\dots,b_j,0,b_{j-2},b_{j-3},\dots,b_i,1,b_{i-1},b_{i-2},\dots,b_0),\\
    &a(i,0,b)+2^j=\\&(b_{z_1-2},b_{z_1-3},\dots,b_j,1,b_{j-2},b_{j-3},\dots,b_i,0,b_{i-1},b_{i-2},\dots,b_0),\\
    &a(i,0,b)+2^i+2^j=\\&(b_{z_1-2},b_{z_1-3},\dots,b_j,1,b_{j-2},b_{j-3},\dots,b_i,1,b_{i-1},b_{i-2},\dots,b_0).
\end{align*}}%
Then by \eqref{eq:dmd}, for $b\in\mathcal{T}_{j-1}$, we have
\begin{align*}
	&\overline{D}'_{2i}(a,b)
	=\\&\left\{\begin{array}{ll}
		L'_{6i+1} - L'_{6j+1} & \mbox{if~} a=a(i,0,b), \\
		\gamma_{6i+1} (L'_{6i+1} - L'_{6j+1}) & \mbox{if~} a=a(i,0,b)+2^i,\\
		\gamma_{6j+1} (L'_{6i+1} - L'_{6j+1}) & \mbox{if~} a=a(i,0,b)+2^j,\\
		\gamma_{6i+1} \gamma_{6j+1} (L'_{6i+1} - L'_{6j+1}) & \mbox{if~} a=a(i,0,b)+2^i+2^j, \\
		\mathbf{0} & \mbox{otherwise},
	\end{array}\right.
\end{align*}
where $\mathbf{0}$ in the last line denotes the all-zero column vector of length $z_1$; the vectors $L'_{6i+1}$ and $L'_{6i+2}$ are defined in \eqref{eq:cvn}. Since $(\lambda_{6i+1}-\lambda_{6j+1})$ is a common factor of all the entries in $(L'_{6i+1} - L'_{6j+1})$, it is also a common factor of all the entries in $\overline{D}'_{2i}(b)$ for every $b\in\mathcal{T}_{j-1}$. The size of $\mathcal{T}_{j-1}$ is $|\mathcal{T}_{j-1}|=2^{z_1-2}$, so we can extract the factor $(\lambda_{6i+1}-\lambda_{6j+1})$ from $2^{z_1-2}$ columns of the matrix $\overline{D}'_{2i}$. Therefore, $(\lambda_{6i+1}-\lambda_{6j+1})^{2^{z_1-2}}$ is a factor of $\det(\overline{Q}')=\det(\overline{Q})$ for every pair of $i\neq j$.

Using the same method we can prove that $(\lambda_{6i+1}-\lambda_{6j+2})^{2^{z_1-2}}$, $(\lambda_{6i+2}-\lambda_{6j+1})^{2^{z_1-2}}$ and $(\lambda_{6i+2}-\lambda_{6j+2})^{2^{z_1-2}}$ are also factors of $\det(\overline{Q})$ for every pair of $i\neq j$. Therefore, the following polynomial 
\begin{equation} \label{eq:h1}
\begin{aligned}
	h_1:= \prod_{i=0}^{z_1-2}\prod_{j=i+1}^{z_1-1} &\Big( (\lambda_{6i+1}-\lambda_{6j+1})^{2^{z_1-2}}(\lambda_{6i+2}-\lambda_{6j+1})^{2^{z_1-2}}\\&(\lambda_{6i+1}-\lambda_{6j+2})^{2^{z_1-2}}(\lambda_{6i+2}-\lambda_{6j+2})^{2^{z_1-2}} \Big)
\end{aligned}
\end{equation}
is a factor of $\det(\overline{Q})$.
It is easy to see that the degree of both $\lambda_{6i+1}$ and $\lambda_{6i+2}$ in the polynomial $h_1$ is $(z_1-1) 2^{z_1-1}$ for all $0\le i\le z_1-1$. Since the degree of $\lambda_{6i+1}$ and $\lambda_{6i+2}$ in the polynomial $\det(\overline{Q})$ cannot exceed $(z_1-1) 2^{z_1-1}$, we conclude that the 
polynomial $h_2:=\det(\overline{Q})/h_1$ does not contain variables $\lambda_{6i+1}$ or $\lambda_{6i+2}$ for any $0\le i\le z_1-1$. In other words, we have obtained a factorization 
\begin{equation} \label{eq:fxt}
\det(\overline{Q})=h_1\cdot h_2 ,
\end{equation}
where $h_1$ only contains $\lambda_{6i+1}$ and  $\lambda_{6i+2}$ for $0\le i\le z_1-1$, and $h_2$ only contains $\gamma_{6i+1}$ and  $\gamma_{6i+2}$ for $0\le i\le z_1-1$.

The next step is to prove that $(\gamma_{6i+1}-\gamma_{6i+2})^{2^{z_1-1}}$ is a factor of $\det(\overline{Q})$ for all $0\le i\le z_1-1$. To that end, we replace the column $\overline{D}_{2i}(b)$ with $\overline{D}_{2i}(b)-\overline{D}_{2i+1}(b)$ for every $0\le b\le 2^{z_1-1}-1$. Clearly, these operations do not change the determinant of the matrix $\overline{Q}$. For every $0\le b\le 2^{z_1-1}-1$, the only two nonzero parts in the column vector $\overline{D}_{2i}(b)-\overline{D}_{2i+1}(b)$ are $L'_{6i+1}-L'_{6i+2}$ and $\gamma_{6i+1}L'_{6i+1}-\gamma_{6i+2}L'_{6i+2}$. When we set $\lambda_{6i+1}=\lambda_{6i+2}$, the vector $L'_{6i+1}-L'_{6i+2}$ becomes the all-zero vector, and the vector $\gamma_{6i+1}L'_{6i+1}-\gamma_{6i+2}L'_{6i+2}$ becomes $(\gamma_{6i+1}-\gamma_{6i+2})L'_{6i+1}$. Therefore, when $\lambda_{6i+1}=\lambda_{6i+2}$, $(\gamma_{6i+1}-\gamma_{6i+2})$ is a common factor of all the entries in $\overline{D}_{2i}(b)-\overline{D}_{2i+1}(b)$. Since $b$ takes $2^{z_1-1}$ possible values, we conclude that $(\gamma_{6i+1}-\gamma_{6i+2})^{2^{z_1-1}}$ is a factor of $\det(\overline{Q})$ when $\lambda_{6i+1}=\lambda_{6i+2}$. On the other hand, the factorization \eqref{eq:fxt} tells us that the factors containing $\gamma_{6i+1}$ and  $\gamma_{6i+2}$ are independent of the the values of $\lambda_{6i+1}$ and  $\lambda_{6i+2}$. Therefore, $(\gamma_{6i+1}-\gamma_{6i+2})^{2^{z_1-1}}$ is always a factor of $\det(\overline{Q})$ no matter $\lambda_{6i+1}$ is equal to $\lambda_{6i+2}$ or not.
Since the degree of $\gamma_{6i+1}$ and $\gamma_{6i+2}$ in the polynomial $\det(\overline{Q})$ cannot exceed $2^{z_1-1}$, we conclude that $(\gamma_{6i+1}-\gamma_{6i+2})^{2^{z_1-1}}$ is the only factor in $\det(\overline{Q})$ that contains $\gamma_{6i+1}$ and $\gamma_{6i+2}$. Therefore, the polynomial $h_2$ in \eqref{eq:fxt} can be written as
\begin{equation} \label{eq:h2}
h_2= c \prod_{i=0}^{z_1-1} (\gamma_{6i+1}-\gamma_{6i+2})^{2^{z_1-1}} ,
\end{equation}
where $c$ is a constant that is independent of $\gamma_{6i+1},\gamma_{6i+2},\lambda_{6i+1},\lambda_{6i+2}$ for all $0\le i\le z_1-1$.

The final step is to prove that the constant $c$ in \eqref{eq:h2} is either $1$ or $-1$. To that end, we write $\det(\overline{Q})$ as a linear combination of monomials, and we will show that 
\begin{equation} \label{eq:mono}
\prod_{i=0}^{z_1-1} \Big( \lambda_{6i+1}^{(z_1-1-i) 2^{z_1-1}}\cdot \lambda_{6i+2}^{(z_1-1-i) 2^{z_1-1}}\cdot \gamma_{6i+1}^{2^{z_1-1}} \Big)
\end{equation}
is a monomial in this linear combination. To see this, we pick the entry $\gamma_{6i+1}\lambda_{6i+1}^{z_1-1-i}$ from the column $\overline{D}_{2i}(b)$, and we pick the entry $\lambda_{6i+2}^{z_1-1-i}$ from the column $\overline{D}_{2i+1}(b)$ for every $0\le i\le z_1-1$ and every $0\le b\le 2^{z_1-1}-1$. Since all these entries are located in different rows, their product, which is given by the monomial in \eqref{eq:mono}, appears in $\det(\overline{Q})$. Moreover, the coefficient of this monomial in $\det(\overline{Q})$ is either $1$ or $-1$. 
%because no other monomials in $\det(\overline{Q})$ have the same degree\footnote{We say that two monomials have the same degree if every variable has the same degree in these two monomials.} as this one.
On the other hand, we can also write $h_1$ and $h_2$ as linear combinations of monomials. In particular, 
$$
\prod_{i=0}^{z_1-1} \Big( \lambda_{6i+1}^{(z_1-1-i) 2^{z_1-1}}\cdot \lambda_{6i+2}^{(z_1-1-i) 2^{z_1-1}} \Big)
$$
is a monomial that appears in $h_1$. To see this, for $0\le i<j\le z_1-1$, we write the product $(\lambda_{6i+1}-\lambda_{6j+1})^{2^{z_1-2}}(\lambda_{6i+2}-\lambda_{6j+1})^{2^{z_1-2}}(\lambda_{6i+1}-\lambda_{6j+2})^{2^{z_1-2}}(\lambda_{6i+2}-\lambda_{6j+2})^{2^{z_1-2}}$ on the right-hand side of \eqref{eq:h1} as a linear combination of monomials. It is easy to see that $\lambda_{6i+1}^{2^{z_1-1}} \lambda_{6i+2}^{2^{z_1-1}}$ is a monomial in this linear combination. Therefore, the product $\prod_{i=0}^{z_1-2}\prod_{j=i+1}^{z_1-1}(\lambda_{6i+1}^{2^{z_1-1}} \lambda_{6i+2}^{2^{z_1-1}})=\prod_{i=0}^{z_1-1} \Big( \lambda_{6i+1}^{(z_1-1-i) 2^{z_1-1}}\cdot \lambda_{6i+2}^{(z_1-1-i) 2^{z_1-1}} \Big)$ appears in $h_1$. By \eqref{eq:h2}, it is easy to see that $c \prod_{i=0}^{z_1-1} \gamma_{6i+1}^{2^{z_1-1}}$ is a monomial that appears in $h_2$. Therefore,
$$
c \prod_{i=0}^{z_1-1} \Big( \lambda_{6i+1}^{(z_1-1-i) 2^{z_1-1}}\cdot \lambda_{6i+2}^{(z_1-1-i) 2^{z_1-1}}\cdot \gamma_{6i+1}^{2^{z_1-1}} \Big)
$$
is a monomial that appears in $\det(\overline{Q})=h_1\cdot h_2$. From the previous analysis, we know that the coefficient $c$ is either $1$ or $-1$.

Since all the $\lambda_i$'s are distinct, we have $h_1\neq 0$. The condition \eqref{eq:cond} further guarantees that $h_2\neq 0$. Thus we conclude that $\det(\overline{Q})\neq 0$, so $Q$ and $\overline{Q}$ are both invertible.
\end{proof}

\begin{figure*}[htbp]
	\begin{equation}\label{eq:add4}
		M_{\vec{z}}={\footnotesize
			\left[\hspace{-0.2cm}\begin{array}{*{2}{@{\hspace*{0.05in}}c}|*{2}{@{\hspace*{0.05in}}c}|*{2}{@{\hspace*{0.05in}}c}|*{2}{@{\hspace*{0.12in}}c}|*{2}{@{\hspace*{0.12in}}c}}
				A_{i_1}^{(2^{z-1})} & & \dots & &  A_{i_{r-2}}^{(2^{z-1})} & & I \otimes L_{6z-6} & I \otimes (L_{6z-6}-L_{6z-5}) & I \otimes L_{6z-4} & \\
				& A_{i_1}^{(2^{z-1})} & & \dots & &  A_{i_{r-2}}^{(2^{z-1})} & & I \otimes L_{6z-5} & I \otimes (L_{6z-3}-L_{6z-4}) & I \otimes L_{6z-3} 
			\end{array}\hspace{-0.2cm}\right]}
	\end{equation}
\end{figure*}

As mentioned at the beginning of Case 1, our objective is to prove that $\vec{y}=\mathbf{0}$ is the only solution to the equation $M_{\vec{z}}~\vec{y}=\mathbf{0}$, where the vector $\vec{y}$ is defined in \eqref{eq:ydb}. Now we have shown that all the coordinates in \eqref{eq:st1} are equal to zero. Replacing all the coordinates in \eqref{eq:st1} with $0$ in the equation $M_{\vec{z}}~\vec{y}=\mathbf{0}$, we obtain
\begin{equation} \label{eq:xcl}
	{\small \begin{aligned}
		\sum_{i=0}^{z_1-1} \Big( (1-b_i) L_{6i} y_{3i,b} + b_i L_{6i+3} y_{3i+1,b} +
		L_{6i+4+b_i} y_{3i+2,b} \Big) = \mathbf{0}&\\ 
		\quad \text{for~} 0\le b\le 2^{z_1}-1,&
	\end{aligned}}
\end{equation}
where $L_i$ is a column vector of length $r$ defined in \eqref{eq:dli}. Note that 
$$
(1-b_i) L_{6i} y_{3i,b} + b_i L_{6i+3} y_{3i+1,b} = \left\{ \begin{array}{ll}
L_{6i} y_{3i,b}  & \mbox{if~} b_i=0, \\
L_{6i+3} y_{3i+1,b}   & \mbox{if~} b_i=1.
\end{array} \right.
$$
Therefore, the linear combination on the left-hand side of \eqref{eq:xcl} contains $2z_1$ distinct $L_i$'s. Since the length of each $L_i$ in this case is $r=3z_1$, these $L_i$'s are linearly independent, so their coefficients must be $0$. This tells us that $y_{3i+2,b}=0$ for all $0\le b\le 2^{z_1}-1$; $y_{3i,b}=0$ if $b_i=0$; $y_{3i+1,b}=0$ if $b_i=1$. Combining this with the fact that all the coordinates in \eqref{eq:st1} are equal to zero, we conclude that $\vec{y}=\mathbf{0}$ is the only solution to the equation $M_{\vec{z}}~\vec{y}=\mathbf{0}$. This completes the proof for Case 1.

\textbf{Case $2$:} $\vec{z}=(z_1,z_2,0,0,0,0,0)$. 
In this case, \eqref{eq:wsmr} implies that $r=3z_1+2z_2,z=z_1+z_2$, and the size of $M_{\vec{z}}$ (defined in \eqref{eq:mzv}) is $2^{z_1+z_2} r\times 2^{z_1+z_2} r$.

Our task is still to show that $M_{\vec{z}}$ is invertible, and we prove it by induction on $z_2$. The base case $z_2=0$ was already proved in Case 1. For the inductive step, we assume that the conclusion holds for $(z_1,z_2-1,0,0,0,0,0)$, and we prove the conclusion for $(z_1,z_2,0,0,0,0,0)$.

To prove the invertibility of $M_{\vec{z}}$, our method is still to show that $\vec{y}=\mathbf{0}$ is the only solution to the equation $M_{\vec{z}}~\vec{y}=\mathbf{0}$, where the vector $\vec{y}$ is defined as
\begin{equation} \label{eq:xxx}
	\begin{aligned}
		\vec{y}=(&y_{0,0},y_{0,1}\cdots ,y_{0,2^z-1},\\&y_{1,0},y_{1,1},\cdots ,y_{1,2^z-1},\\&\cdots,\\&y_{r-1,0},y_{r-1,1},\cdots,y_{r-1,2^z-1})^T .
	\end{aligned}
\end{equation}
Recall from \eqref{eq:mzv} that $i_1,i_2,\dots,i_r$ are the indices in the set $\mathcal{F}(\vec{z})$. Since we assume that $\vec{z}=(z_1,z_2,0,0,0,0,0)$ with $z_2\ge 1$, we have $i_{r-1}=3z-3$ and $i_r=3z-2$.
Note that $M_{\vec{z}}$ is matrix \eqref{eq:add4}, 
where $\otimes$ is the Kronecker product, $I=I_{2^{z-1}}$ is the identity matrix of size $2^{z-1}\times 2^{z-1}$, and the column vector $L_i$ is defined in \eqref{eq:dli}.
%$${\tiny
%\left[\hspace{-0.2cm}\begin{array}{*{2}{@{\hspace*{0.05in}}l}|*{2}{@{\hspace*{0.05in}}l}|*{2}{@{\hspace*{0.05in}}l}|*{2}{@{\hspace*{0.05in}}l}|*{2}{@{\hspace*{0.05in}}l}}
%    A_{i_1}^{(2^{z-1})} & & \dots & &  A_{i_{r-2}}^{(2^{z-1})} & & I_{2^{z-1}} \otimes L_{6z-6} & I_{2^{z-1}} \otimes (L_{6z-6}-L_{6z-5}) & I_{2^{z-1}} \otimes L_{6z-4} & \\
%     & A_{i_1}^{(2^{z-1})} & & \dots & &  A_{i_{r-2}}^{(2^{z-1})} & & I_{2^{z-1}} \otimes L_{6z-5} & I_{2^{z-1}} \otimes (L_{6z-3}-L_{6z-4}) & I_{2^{z-1}} \otimes L_{6z-3} 
%\end{array}\hspace{-0.2cm}\right].}
%$$
We define a $2^{z-1}r \times 2^{z-1}r$ matrix
{\tiny\begin{align*}
	 Q_{\vec{z}} = \left[ A_{i_1}^{(2^{z-1})} \quad A_{i_2}^{(2^{z-1})} \quad \dots \quad A_{i_{r-2}}^{(2^{z-1})} \quad I_{2^{z-1}} 
		\otimes L_{6z-6}  \quad I_{2^{z-1}} 
		\otimes L_{6z-3} \right],
\end{align*}}%
%where $\otimes$ is the Kronecker product, $I_{2^{z-1}}$ is the identity matrix of size $2^{z-1}\times 2^{z-1}$, and the column vector $L_i$ is defined in \eqref{eq:dli}.
and a vector 
\begin{align*}
	\vec{y}^{\text{sum}}=(&y_{0,0}^{\text{sum}},y_{0,1}^{\text{sum}}\cdots ,y_{0,2^{z-1}-1}^{\text{sum}},\\&y_{1,0}^{\text{sum}},y_{1,1}^{\text{sum}},\cdots ,y_{1,2^{z-1}-1}^{\text{sum}},\\&\cdots,\\&y_{r-1,0}^{\text{sum}},y_{r-1,1}^{\text{sum}},\cdots,y_{r-1,2^{z-1}-1}^{\text{sum}})^T ,
\end{align*}
where $y_{i,a}^{\text{sum}}=y_{i,a}+y_{i,a+2^{z-1}}$ for $0\le i\le r-1$ and $0\le a\le 2^{z-1}-1$. 
If we add the $a$th block row of the equation $M_{\vec{z}}~\vec{y}=\mathbf{0}$ to the $(a+2^{z-1})$th block row for every $0\le a\le 2^{z-1}-1$, then we obtain $Q_{\vec{z}} \vec{y}^{\text{sum}}=\mathbf{0}$.

The next step is to show that $Q_{\vec{z}}$ is invertible, which implies that $\vec{y}^{\text{sum}}$ is an all-zero vector.

\begin{lemma} \label{lm:idt}
If $M_{(z_1,z_2-1,0,0,0,0,0)}$ is invertible, then $Q_{(z_1,z_2,0,0,0,0,0)}$ is also invertible.
\end{lemma}
The condition of this lemma is the induction hypothesis.
\begin{proof}
Let us write $\vec{z}=(z_1,z_2,0,0,0,0,0)$ and $\vec{z}'=(z_1,z_2-1,0,0,0,0,0)$. Note that the size of $M_{\vec{z}'}$ is $2^{z-1}(r-2)\times 2^{z-1}(r-2)$, where $z=z_1+z_2$ and $r=3z_1+2z_2$. In order to prove that $Q_{\vec{z}}$ is invertible, we will show that $\vec{y}^{\text{sum}}=\mathbf{0}$ is the only solution to the equation $Q_{\vec{z}} \vec{y}^{\text{sum}}=\mathbf{0}$.

We define a quadratic polynomial $f(x)=(x-\lambda_{6z-6})(x-\lambda_{6z-3})$, and we write its coefficients as $f_0,f_1,f_2$, i.e., we write $f(x)=f_0+f_1x+f_2x^2$.
%For every $0\le a\le 2^{z-1}-1$, we define an $(r-2)\times r$ matrix $F_a$ as
We define an $(r-2)\times r$ matrix $F_0$ as
$$
F_0 := \begin{bmatrix}
f_0&f_1&f_2&0&0&0&0&0&\cdots&0\\
    0&f_0&f_1&f_2&0&0&0&0&\cdots&0\\
    0&0&f_0&f_1&f_2&0&0&0&\cdots&0\\
\vdots&\vdots&\vdots&\vdots&\vdots&\vdots&\vdots&\vdots&\vdots&\vdots\\
    0&0&0&&0&\cdots&0&f_0&f_1&f_2
\end{bmatrix},
$$
%By definition, $F_0=F_1=\dots=F_{2^{z-1}-1}$.
and a $2^{z-1} (r-2)\times 2^{z-1} r$ matrix $F$ as
$$
F := \begin{bmatrix}
    F_0 & 0 & 0 & \cdots & 0 \\
    0 & F_0 & 0 & \cdots & 0 \\
    0 & 0 & F_0 & \cdots & 0 \\
    \vdots & \vdots & \vdots & \vdots & \vdots \\
    0 & 0 & 0 & \cdots & F_0 
\end{bmatrix} ,
$$
where each $0$ is a $(r-2)\times r$ matrix with all zero entries, and $F_0$ appears in total $2^{z-1}$ times  on the diagonal.

The equation $Q_{\vec{z}} \vec{y}^{\text{sum}}=\mathbf{0}$ implies that $F Q_{\vec{z}} \vec{y}^{\text{sum}}=\mathbf{0}$. It is easy to see that
$$
F(I_{2^{z-1}} 
\otimes L_{6z-6}) = F(I_{2^{z-1}} 
\otimes L_{6z-3}) = 0,
$$
where $0$ on the right-hand side is the all-zero matrix of size $2^{z-1} (r-2)\times 2^{z-1}$. Therefore, 
\begin{equation} \label{eq:fq1}
F Q_{\vec{z}}= \left[ F A_{i_1}^{(2^{z-1})} \quad F A_{i_2}^{(2^{z-1})} \quad \dots \quad F A_{i_{r-2}}^{(2^{z-1})} \quad 0 \quad 0 \right] .
\end{equation}
Define $r-2$ matrices 
\begin{equation} \label{eq:fq2}
B_{i_1}= F A_{i_1}^{(2^{z-1})}, B_{i_2}= F A_{i_2}^{(2^{z-1})}, \dots, B_{i_{r-2}}= F A_{i_{r-2}}^{(2^{z-1})}
\end{equation}
of size $2^{z-1} (r-2)\times 2^{z-1}$. We further define a $2^{z-1} (r-2)\times 2^{z-1}(r-2)$ matrix
\begin{equation} \label{eq:fq3}
K=[B_{i_1} \quad B_{i_2} \quad \cdots \quad B_{i_{r-2}}] .
\end{equation}
The next step is to show that $K$ is invertible. To that end, we represent each matrix $B_{i_t}$ as
\begin{equation} \label{eq:hhh}
{\scriptsize \begin{bmatrix}
			B_{i_t}(0,0) & B_{i_t}(0,1) & \cdots & B_{i_t}(0,2^{z-1}-1) \\
			B_{i_t}(1,0) & B_{i_t}(1,1) & \cdots & B_{i_t}(1,2^{z-1}-1) \\
			\vdots & \vdots & \vdots & \vdots \\
			B_{i_t}(2^{z-1}-1,0) & B_{i_t}(2^{z-1}-1,1) & \cdots & B_{i_t}(2^{z-1}-1,2^{z-1}-1) 
		\end{bmatrix}}
\end{equation}
for $1\le t\le r-2$, where each $B_{i_t}(a,b)$ is a column vector of length $r-2$ for $0\le a,b\le 2^{z-1}-1$. In order to characterize the matrices $B_{i_1},\dots,B_{i_{r-2}}$, we further introduce the notation 
$$
\overline{L}_i := \begin{bmatrix}
1 \\
\lambda_i \\
\lambda_i^2 \\
\vdots \\
\lambda_i^{r-3}
\end{bmatrix} .
$$
By definition, if $i_t=3j$ for some $j\in\{0,1,\dots,z-2\}$, then
\begin{align*}
	&B_{i_t}(a,b)\\
	=&B_{3j}(a,b) \\
	=&\left\{
	\begin{array}{ll}
		f(\lambda_{6j+a_j}) \overline{L}_{6j+a_j}   & \mbox{if~} a=b, \\
		f(\lambda_{6j})\overline{L}_{6j}-f(\lambda_{6j+1})\overline{L}_{6j+1}   & \mbox{if~} a_j=0, b_j=1 \\&\mbox{and~} a_i=b_i\, \forall i\neq j, \\
		\mathbf{0} & \mbox{otherwise}.
	\end{array}
	\right.
\end{align*}
If $i_t=3j+1$ for some $j\in\{0,1,\dots,z-2\}$, then
\begin{align*}
	&B_{i_t}(a,b)\\
	=&B_{3j+1}(a,b) \\
	=&\left\{
	\begin{array}{ll}
		f(\lambda_{6j+2+a_j}) \overline{L}_{6j+2+a_j}   & \mbox{if~} a=b, \\
		f(\lambda_{6j+3})\overline{L}_{6j+3}-f(\lambda_{6j+2})\overline{L}_{6j+2}   & \mbox{if~} a_j=1, b_j=0 \\&\mbox{and~} a_i=b_i\, \forall i\neq j, \\
		\mathbf{0} & \mbox{otherwise}.
	\end{array}
	\right.
\end{align*}
If $i_t=3j+2$ for some $j\in\{0,1,\dots,z-2\}$, then
\begin{align*}
	&B_{i_t}(a,b)\\
	=&B_{3j+2}(a,b) \\
	=&\left\{
	\begin{array}{ll}
		f(\lambda_{6j+4+a_j}) \overline{L}_{6j+4+a_j}   & \mbox{if~} a=b, \\
		\mathbf{0} & \mbox{otherwise}.
	\end{array}
	\right.
\end{align*}

Next we define another $r-2$ matrices $\overline{B}_{i_1}, \overline{B}_{i_2}, \dots, \overline{B}_{i_{r-2}}$ whose size is the same as $B_{i_1},B_{i_2},\dots,B_{i_{r-2}}$, and
$$
\overline{K}=[\overline{B}_{i_1} \quad \overline{B}_{i_2} \quad \cdots \quad \overline{B}_{i_{r-2}}] .
$$
We define the column vector $\overline{B}_{i_t}(a,b)$ in the same way as we defined $B_{i_t}(a,b)$ in \eqref{eq:hhh}. The value of $\overline{B}_{i_t}(a,b)$ is given as follows: If $i_t=3j$ for some $j\in\{0,1,\dots,z-2\}$, then
\begin{align*}
	&\overline{B}_{i_t}(a,b)\\
	=&\overline{B}_{3j}(a,b) \\
	=&\left\{
	\begin{array}{ll}
		f(\lambda_{6j+a_j}) \overline{L}_{6j+a_j}   & \mbox{if~} a=b, \\
		-f(\lambda_{6j+1})\overline{L}_{6j+1}   & \mbox{if~} a_j=0, b_j=1 \\&\mbox{and~} a_i=b_i\, \forall i\neq j, \\
		\mathbf{0} & \mbox{otherwise}.
	\end{array}
	\right.
\end{align*}
If $i_t=3j+1$ for some $j\in\{0,1,\dots,z-2\}$, then
\begin{align*}
	&\overline{B}_{i_t}(a,b)\\
	=&\overline{B}_{3j+1}(a,b) \\
	=&\left\{
	\begin{array}{ll}
		f(\lambda_{6j+2+a_j}) \overline{L}_{6j+2+a_j}   & \mbox{if~} a=b, \\
		-f(\lambda_{6j+2})\overline{L}_{6j+2}   & \mbox{if~} a_j=1, b_j=0 \\&\mbox{and~} a_i=b_i\, \forall i\neq j, \\
		\mathbf{0} & \mbox{otherwise}.
	\end{array}
	\right.
\end{align*}
If $i_t=3j+2$ for some $j\in\{0,1,\dots,z-2\}$, then
\begin{align*}
	&\overline{B}_{i_t}(a,b)\\
	=&\overline{B}_{3j+2}(a,b) \\
	=&\left\{
	\begin{array}{ll}
		f(\lambda_{6j+4+a_j}) \overline{L}_{6j+4+a_j}   & \mbox{if~} a=b, \\
		\mathbf{0} & \mbox{otherwise}.
	\end{array}
	\right.
\end{align*}
We use $B_{i_t}(b)$ and $\overline{B}_{i_t}(b)$ to denote the $b$th column of $B_{i_t}$ and $\overline{B}_{i_t}$, respectively. The relation between the matrices $B_{i_t}$ and $\overline{B}_{i_t}$ can be described as follows: If $i_t=3j$ for some $j\in\{0,1,\dots,z-2\}$, then 
$$
\overline{B}_{i_t}(b) = \left\{
\begin{array}{ll}
 B_{i_t}(b) - B_{i_t}(b-2^j)  & \text{~if~} b_j=1, \\
 B_{i_t}(b) & \text{~if~} b_j=0.
\end{array}
\right.
$$
If $i_t=3j+1$ for some $j\in\{0,1,\dots,z-2\}$, then
$$
\overline{B}_{i_t}(b) = \left\{
\begin{array}{ll}
 B_{i_t}(b) - B_{i_t}(b+2^j) & \text{~if~} b_j=0, \\
 B_{i_t}(b) & \text{~if~} b_j=1.
\end{array}
\right.
$$
If $i_t=3j+2$ for some $j\in\{0,1,\dots,z-2\}$, then $\overline{B}_{i_t} = B_{i_t}$. The relation between $B_{i_t}$ and $\overline{B}_{i_t}$ immediately implies that $\det(K)=\det(\overline{K})$. Therefore, $K$ is invertible if and only if $\overline{K}$ is invertible.

To prove the invertibility of $\overline{K}$, we introduce another variation of $B_{i_t}$ and $K$:
Define $r-2$ matrices $\widetilde{B}_{i_1}, \widetilde{B}_{i_2}, \dots, \widetilde{B}_{i_{r-2}}$ whose size is the same as $B_{i_1},B_{i_2},\dots,B_{i_{r-2}}$.
We further define
$$
\widetilde{K}=[\widetilde{B}_{i_1} \quad \widetilde{B}_{i_2} \quad \cdots \quad \widetilde{B}_{i_{r-2}}] .
$$
The column vector $\widetilde{B}_{i_t}(a,b)$ is defined in the same way as we defined $B_{i_t}(a,b)$ in \eqref{eq:hhh}. The value of $\widetilde{B}_{i_t}(a,b)$ is given as follows: If $i_t=3j$ for some $j\in\{0,1,\dots,z-2\}$, then
$$
\widetilde{B}_{i_t}(a,b)=\widetilde{B}_{3j}(a,b) =\left\{
\begin{array}{ll}
 \overline{L}_{6j+a_j}   & \mbox{if~} a=b, \\
  -\overline{L}_{6j+1}   & \mbox{if~} a_j=0, b_j=1 \\&\mbox{and~} a_i=b_i\, \forall i\neq j, \\
  \mathbf{0} & \mbox{otherwise}.
\end{array}
\right.
$$
If $i_t=3j+1$ for some $j\in\{0,1,\dots,z-2\}$, then
$$
\widetilde{B}_{i_t}(a,b)=\widetilde{B}_{3j+1}(a,b) =\left\{
\begin{array}{ll}
 \overline{L}_{6j+2+a_j}   & \mbox{if~} a=b, \\
  -\overline{L}_{6j+2}   & \mbox{if~} a_j=1, b_j=0 \\&\mbox{and~} a_i=b_i\, \forall i\neq j,  \\
  \mathbf{0} & \mbox{otherwise}.
\end{array}
\right.
$$
If $i_t=3j+2$ for some $j\in\{0,1,\dots,z-2\}$, then
$$
\widetilde{B}_{i_t}(a,b)=\widetilde{B}_{3j+2}(a,b) =\left\{
\begin{array}{ll}
 \overline{L}_{6j+4+a_j}   & \mbox{if~} a=b, \\
  \mathbf{0} & \mbox{otherwise}.
\end{array}
\right.
$$
We use $\widetilde{B}_{i_t}(b)$ to denote the $b$th column of $\widetilde{B}_{i_t}$. The relation between $\overline{B}_{i_t}$ and $\widetilde{B}_{i_t}$ can be described as follows: If $i_t=3j$ for some $j\in\{0,1,\dots,z-2\}$, then
$$
\overline{B}_{i_t}(b) = f(\lambda_{6j+b_j}) \widetilde{B}_{i_t}(b) .
$$
If $i_t=3j+1$ for some $j\in\{0,1,\dots,z-2\}$, then
$$
\overline{B}_{i_t}(b) = f(\lambda_{6j+2+b_j}) \widetilde{B}_{i_t}(b) .
$$
If $i_t=3j+2$ for some $j\in\{0,1,\dots,z-2\}$, then
$$
\overline{B}_{i_t}(b) = f(\lambda_{6j+4+b_j}) \widetilde{B}_{i_t}(b) .
$$
Since all the $\lambda_i$'s are distinct, the coefficients $f(\lambda_{6j+b_j}),f(\lambda_{6j+2+b_j}),f(\lambda_{6j+4+b_j})$ in the above equations are nonzero. Therefore, $\widetilde{B}_{i_t}$ is obtained from multiplying a nonzero element to each column of $\overline{B}_{i_t}$ for every $1\le t\le r-2$. Thus we conclude that the invertibility of $\widetilde{K}$ is equivalent to the invertibility of $\overline{K}$, which is further equivalent to the invertibility of $K$.

Next we use the invertibility of $M_{\vec{z}'}$ to prove the invertibility of $\widetilde{K}$. By definition \eqref{eq:mzv},
$$
M_{\vec{z}'} = \left[ A_{i_1}^{(2^{z-1})} \quad A_{i_2}^{(2^{z-1})} \quad \dots \quad A_{i_{r-2}}^{(2^{z-1})} \right] .
$$
For $0\le b\le 2^{z-1}-1$ and $1\le t\le r-2$, we use $A_{i_t}^{(2^{z-1})}(b)$ to denote the $b$th column of $A_{i_t}^{(2^{z-1})}$. It is easy to verify the following relation between $A_{i_t}^{(2^{z-1})}$ and $\widetilde{B}_{i_t}$: If $i_t=3j$ for some $j\in\{0,1,\dots,z-2\}$, then 
$$\widetilde{B}_{i_t}(b)=\left\{
\begin{array}{ll}
  A_{i_t}^{(2^{z-1})}(b) - A_{i_t}^{(2^{z-1})}(b-2^j) & \text{~if~} b_j=1, \\
  A_{i_t}^{(2^{z-1})}(b) & \text{~if~} b_j=0..
\end{array}
\right.
$$
If $i_t=3j+1$ for some $j\in\{0,1,\dots,z-2\}$, then
$$
\widetilde{B}_{i_t}(b) = \left\{
\begin{array}{ll}
 A_{i_t}^{(2^{z-1})}(b) - A_{i_t}^{(2^{z-1})}(b+2^j) & \text{~if~} b_j=0, \\
 A_{i_t}^{(2^{z-1})}(b) & \text{~if~} b_j=1.
\end{array}
\right.
$$
If $i_t=3j+2$ for some $j\in\{0,1,\dots,z-2\}$, then $\widetilde{B}_{i_t} = A_{i_t}^{(2^{z-1})}$. The relation between $A_{i_t}^{(2^{z-1})}$ and $\widetilde{B}_{i_t}$ immediately implies that $\det(M_{\vec{z}'})=\det(\widetilde{K})$. Since $M_{\vec{z}'}$ is invertible, we conclude that the three matrices $K,\overline{K}$ and $\widetilde{K}$ are all invertible.

Since $K$ is invertible, the equation $F Q_{\vec{z}} \vec{y}^{\text{sum}}=\mathbf{0}$ and \eqref{eq:fq1}--\eqref{eq:fq3} together imply that the first $2^{z-1}(r-2)$ coordinates of $\vec{y}^{\text{sum}}$ are all equal to $0$. Taking this result back into the equation $Q_{\vec{z}} \vec{y}^{\text{sum}}=\mathbf{0}$, we immediately obtain that the last $2^z$ coordinates of $\vec{y}^{\text{sum}}$ are also equal to $0$. Therefore, $\vec{y}^{\text{sum}}=\mathbf{0}$ is the only solution to the equation $Q_{\vec{z}} \vec{y}^{\text{sum}}=\mathbf{0}$. This completes the proof of the lemma.
\end{proof}

Now we have shown that $\vec{y}^{\text{sum}}$ is an all-zero vector. Before proceeding to prove that $\vec{y}$ is also an all-zero vector, we need to introduce some notation. We define a $2^{z-1}r \times 2^{z-1}r$ matrix
{\tiny $$
Q_{\vec{z}}^{(1)} = \left[ A_{i_1}^{(2^{z-1})} \quad A_{i_2}^{(2^{z-1})} \quad \dots \quad A_{i_{r-2}}^{(2^{z-1})} \quad I_{2^{z-1}} 
\otimes L_{6z-5}  \quad I_{2^{z-1}} 
\otimes L_{6z-4} \right] 
$$}%
and a vector 
\begin{align*}
	\vec{y}^{(1)}=(&y_{0,0}^{(1)},y_{0,1}^{(1)}\cdots ,y_{0,2^{z-1}-1}^{(1)},\\&y_{1,0}^{(1)},y_{1,1}^{(1)},\cdots ,y_{1,2^{z-1}-1}^{(1)},\\&\cdots,\\&y_{r-1,0}^{(1)},y_{r-1,1}^{(1)},\cdots,y_{r-1,2^{z-1}-1}^{(1)})^T ,
\end{align*}
where $y_{i,a}^{(1)}=y_{i,a}$ for all $0\le i\le r-3$ and all $0\le a\le 2^{z-1}-1$; $y_{r-2,a}^{(1)}=-y_{r-2,a+2^{z-1}}$ for all $0\le a\le 2^{z-1}-1$; $y_{r-1,a}^{(1)}=y_{r-1,a}$ for all $0\le a\le 2^{z-1}-1$. 

The result $\vec{y}^{\text{sum}}=\mathbf{0}$ implies that $y_{r-2,a}^{\text{sum}}=y_{r-2,a}+y_{r-2,a+2^{z-1}}=0$ for all $0\le a\le 2^{z-1}-1$. Taking this into the equation $M_{\vec{z}}~\vec{y}=\mathbf{0}$, the first $2^{z-1}$ block rows of $M_{\vec{z}}~\vec{y}=\mathbf{0}$ become $Q_{\vec{z}}^{(1)} \vec{y}^{(1)}=\mathbf{0}$. 
Since the matrices $Q_{\vec{z}}^{(1)}$ and $Q_{\vec{z}}$ have the same structure, we can use the method in the proof of Lemma~\ref{lm:idt} to show that $Q_{\vec{z}}^{(1)}$ is also invertible. Therefore, $\vec{y}^{(1)}=\mathbf{0}$. Combining $\vec{y}^{(1)}=\mathbf{0}$ with $\vec{y}^{\text{sum}}=\mathbf{0}$, we immediately conclude that $\vec{y}$ is an all-zero vector. This proves that $M_{\vec{z}}$ is invertible for Case 2.

\textbf{Case $3$:} $\vec{z}=(z_1,z_2,z_3,0,0,0,0)$. In this case, \eqref{eq:wsmr} implies that $r=3z_1+2z_2+2z_3$ and $z=z_1+z_2+z_3$.
We prove the invertibility of $M_{\vec{z}}$ by induction on $z_3$. The base case $z_3=0$ was already proved in Case 2. For the inductive step, we assume that the conclusion holds for $(z_1,z_2,z_3-1,0,0,0,0)$, and we prove the conclusion for $(z_1,z_2,z_3,0,0,0,0)$.

Our method is still to show that $\vec{y}=\mathbf{0}$ is the only solution to the equation $M_{\vec{z}}~\vec{y}=\mathbf{0}$, where the vector $\vec{y}$ is defined in \eqref{eq:xxx}.
Since we assume that $\vec{z}=(z_1,z_2,z_3,0,0,0,0)$ with $z_3\ge 1$, we have $i_{r-1}=3z-3$ and $i_r=3z-1$ in \eqref{eq:mzv}. We define two $2^{z-1}r \times 2^{z-1}r$ matrices
{\tiny $$
Q_{\vec{z}}^{(2)} = \left[ A_{i_1}^{(2^{z-1})} \quad A_{i_2}^{(2^{z-1})} \quad \dots \quad A_{i_{r-2}}^{(2^{z-1})} \quad I_{2^{z-1}} 
\otimes L_{6z-6}  \quad I_{2^{z-1}} 
\otimes L_{6z-2} \right] , $$
$$
Q_{\vec{z}}^{(3)} = \left[ A_{i_1}^{(2^{z-1})} \quad A_{i_2}^{(2^{z-1})} \quad \dots \quad A_{i_{r-2}}^{(2^{z-1})} \quad I_{2^{z-1}} 
\otimes L_{6z-5}  \quad I_{2^{z-1}} 
\otimes L_{6z-1} \right] .
$$}%
Using the method in the proof of Lemma~\ref{lm:idt}, we can show that both $Q_{\vec{z}}^{(2)}$ and $Q_{\vec{z}}^{(3)}$ are invertible.
We further define two vectors
\begin{align*}
\vec{y}^{(2)}=(&y_{0,0}^{(2)},y_{0,1}^{(2)}\cdots ,y_{0,2^{z-1}-1}^{(2)},\\&y_{1,0}^{(2)},y_{1,1}^{(2)},\cdots ,y_{1,2^{z-1}-1}^{(2)},\\&\cdots,\\&y_{r-1,0}^{(2)},y_{r-1,1}^{(2)},\cdots,y_{r-1,2^{z-1}-1}^{(2)})^T , \\~\\
\vec{y}^{(3)}=(&y_{0,0}^{(3)},y_{0,1}^{(3)}\cdots ,y_{0,2^{z-1}-1}^{(3)},\\&y_{1,0}^{(3)},y_{1,1}^{(3)},\cdots ,y_{1,2^{z-1}-1}^{(3)},\\&\cdots,\\&y_{r-1,0}^{(3)},y_{r-1,1}^{(3)},\cdots,y_{r-1,2^{z-1}-1}^{(3)})^T  ,
\end{align*}
where $y_{i,a}^{(2)}=y_{i,a}$ and $y_{i,a}^{(3)}=y_{i,a+2^{z-1}}$ for $0\le i\le r-1$ and $0\le a\le 2^{z-1}-1$. The last $2^{z-1}$ block rows of $M_{\vec{z}}~\vec{y}=\mathbf{0}$ give us $Q_{\vec{z}}^{(3)} \vec{y}^{(3)}=\mathbf{0}$. Since $Q_{\vec{z}}^{(3)}$ is invertible, we have $\vec{y}^{(3)}=\mathbf{0}$. Taking this back into $M_{\vec{z}}~\vec{y}=\mathbf{0}$, the first $2^{z-1}$ block rows of $M_{\vec{z}}~\vec{y}=\mathbf{0}$ become $Q_{\vec{z}}^{(2)} \vec{y}^{(2)}=\mathbf{0}$. Since $Q_{\vec{z}}^{(2)}$ is invertible, we conclude that $\vec{y}^{(2)}=\mathbf{0}$, so $\vec{y}$ is an all-zero vector. This completes the proof for Case 3.

\textbf{Case $4$:} $\vec{z}=(z_1,z_2,z_3,z_4,0,0,0)$. In this case, \eqref{eq:wsmr} implies that $r=3z_1+2z_2+2z_3+2z_4$ and $z=z_1+z_2+z_3+z_4$.
We prove the invertibility of $M_{\vec{z}}$ by induction on $z_4$. The base case $z_4=0$ was already proved in Case 3. For the inductive step, we assume that the conclusion holds for $(z_1,z_2,z_3,z_4-1,0,0,0)$, and we prove the conclusion for $(z_1,z_2,z_3,z_4,0,0,0)$.

Our method is still to show that $\vec{y}=\mathbf{0}$ is the only solution to the equation $M_{\vec{z}}~\vec{y}=\mathbf{0}$, where the vector $\vec{y}$ is defined in \eqref{eq:xxx}.
Since we assume that $\vec{z}=(z_1,z_2,z_3,z_4,0,0,0)$ with $z_4\ge 1$, we have $i_{r-1}=3z-2$ and $i_r=3z-1$ in \eqref{eq:mzv}. We define two $2^{z-1}r \times 2^{z-1}r$ matrices
{\tiny \begin{align*}
Q_{\vec{z}}^{(4)} = \left[ A_{i_1}^{(2^{z-1})} \quad A_{i_2}^{(2^{z-1})} \quad \dots \quad A_{i_{r-2}}^{(2^{z-1})} \quad I_{2^{z-1}} 
\otimes L_{6z-4}  \quad I_{2^{z-1}} 
\otimes L_{6z-2} \right] , \\
Q_{\vec{z}}^{(5)} = \left[ A_{i_1}^{(2^{z-1})} \quad A_{i_2}^{(2^{z-1})} \quad \dots \quad A_{i_{r-2}}^{(2^{z-1})} \quad I_{2^{z-1}} 
\otimes L_{6z-3}  \quad I_{2^{z-1}} 
\otimes L_{6z-1} \right] .
\end{align*}}%
Using the method in the proof of Lemma~\ref{lm:idt}, we can show that both $Q_{\vec{z}}^{(4)}$ and $Q_{\vec{z}}^{(5)}$ are invertible. Recall the definitions of $\vec{y}^{(2)}$ and $\vec{y}^{(3)}$ in Case 3.
The first $2^{z-1}$ block rows of $M_{\vec{z}}~\vec{y}=\mathbf{0}$ give us $Q_{\vec{z}}^{(4)} \vec{y}^{(2)}=\mathbf{0}$. Since $Q_{\vec{z}}^{(4)}$ is invertible, we have $\vec{y}^{(2)}=\mathbf{0}$. Taking this back into $M_{\vec{z}}~\vec{y}=\mathbf{0}$, the last $2^{z-1}$ block rows of $M_{\vec{z}}~\vec{y}=\mathbf{0}$ become $Q_{\vec{z}}^{(5)} \vec{y}^{(3)}=\mathbf{0}$. Since $Q_{\vec{z}}^{(5)}$ is invertible, we conclude that $\vec{y}^{(3)}=\mathbf{0}$, so $\vec{y}$ is an all-zero vector. This completes the proof for Case 4.

\textbf{Case $5$:} $\vec{z}=(z_1,z_2,z_3,z_4,z_5,0,0)$. The proof is the same as Case 3.

\textbf{Case $6$:} $\vec{z}=(z_1,z_2,z_3,z_4,z_5,z_6,0)$. The proof is the same as Case 4.

\textbf{Case $7$:} $\vec{z}=(z_1,z_2,z_3,z_4,z_5,z_6,z_7)$. The proof is the same as Case 3 and Case 4.

\section{Optimal repair bandwidth for single node failure}\label{repair}
According to Appendix, no matter which node fails, we can always convert it to the first group, so we only need to prove that the nodes in the first group can be repaired. Next, We will use three cases to illustrate the repair procedure.

{\bf First case: How to repair $C_0$.} Note that the parity check equations in \eqref{eq:pc} can be written in the matrix form 
\begin{equation} \label{eq:repair}
A_0 C_0+A_1 C_1+A_2 C_2+\dots+A_{n-1} C_{n-1}=\mathbf{0} ,
\end{equation}
where each $C_i$ is a column vector of length $\ell$. Each block row in the matrices $A_0,\dots,A_{n-1}$ corresponds to a set of $r$ parity check equations because the length of each $L_i$ is $r$. Since there are $\ell$ block rows in each matrix $A_i$, we have $\ell$ sets of parity check equations in total. The repair of $C_0$ only involves $\ell/2$ out of these $\ell$ sets of parity check equations. More precisely, among the $\ell$ block rows in each matrix $A_i$, we only need to look at the block rows whose indices lie in the set $\{0,2,4,\cdots,2(\ell/2-1)\}$. 
These $\ell/2$ block rows of parity check equations can again be organized in the matrix form
\begin{equation} \label{eq:repair0}
\widetilde{A}_0 \widetilde{C}_0 + \widehat{A}_0 \widehat{C}_0
+ \sum_{i=1}^{n-1} \overline{A}_i \overline{C}_i = \mathbf{0} ,
\end{equation}
where $\widetilde{A}_0, \widehat{A}_0, \overline{A}_i, 1\le i\le n-1$ are all $\ell/2\times \ell/2$ matrices, and $\widetilde{C}_0,\widehat{C}_0,\overline{C}_i, 1\le i\le n-1$ are all column vectors of length $\ell/2$. More specifically, the matrices in \eqref{eq:repair0} are for every $0\le a,b\le \ell/2-1$,
\begin{equation*}
\begin{aligned}
\widetilde{A}_{0}(a,b) &=\left\{
\begin{array}{ll}
  L_0   & \mbox{if~} a=b, \\
  \mathbf{0} & \mbox{otherwise},
\end{array}
\right.  \\
\widehat{A}_{0}(a,b) &=\left\{
\begin{array}{ll}
  -L_1   & \mbox{if~} a=b, \\
  \mathbf{0} & \mbox{otherwise},
\end{array}
\right.  \\
\overline{A}_{1}(a,b) &=\left\{
\begin{array}{ll}
  L_2   & \mbox{if~} a=b, \\
  \mathbf{0} & \mbox{otherwise},
\end{array}
\right.  \\
\overline{A}_{2}(a,b) &=\left\{
\begin{array}{ll}
  L_4   & \mbox{if~} a=b, \\
  \mathbf{0} & \mbox{otherwise},
\end{array}
\right.  
\end{aligned}
\end{equation*}
and for every $1\le i\le n/3-1$,
\begin{equation*}
\begin{aligned}
\overline{A}_{3i}(a,b) &=\left\{
\begin{array}{ll}
  L_{6i+a_{i-1}}   & \mbox{if~} a=b, \\
  L_{6i}-L_{6i+1}   & \mbox{if~} a_{i-1}=0, b_{i-1}=1, \\&\mbox{and~} a_j=b_j \,\forall j\neq i-1, \\
  \mathbf{0} & \mbox{otherwise},
\end{array}
\right.  \\
\overline{A}_{3i+1}(a,b) &=\left\{
\begin{array}{ll}
  L_{6i+2+a_{i-1}}   & \mbox{if~} a=b, \\
  L_{6i+3}-L_{6i+2}   & \mbox{if~} a_{i-1}=1, b_{i-1}=0, \\&\mbox{and~} a_j=b_j \,\forall j\neq i-1, \\
  \mathbf{0} & \mbox{otherwise},
\end{array}
\right. \\
\overline{A}_{3i+2}(a,b) &=\left\{
\begin{array}{ll}
  L_{6i+4+a_{i-1}}   & \mbox{if~} a=b, \\
  \mathbf{0} & \mbox{otherwise},
\end{array}
\right.
\end{aligned}
\end{equation*}
where $\mathbf{0}$ in the last line denotes the all-zero column vector of length $r$, 
and the column vectors in \eqref{eq:repair0} are for every $0\le a\le \ell/2-1$,
\begin{equation*}
\begin{aligned}
\widetilde{C}_{0}(a) &= C_0(2a)+C_0(2a+1),\\
\widehat{C}_{0}(a) &= C_0(2a+1), 
\end{aligned}
\end{equation*}
and for every $0\le i\le n-1$ and $i\neq 0$,
$$\overline{C}_i(a) =C_i(2a).$$
Here we make an important observation: The matrices $\overline{A}_3,\overline{A}_4,\dots,\overline{A}_{n-1}$ are precisely the $n-3$ parity check matrices that would appear in our MSR code construction with code length $n-3$ and subpacketization $\ell/2=2^{n/3-1}=2^{(n-3)/3}$. The other $4$ matrices $\widetilde{A}_0,\widehat{A}_0,\overline{A}_1,\overline{A}_2$ are all block-diagonal matrices, and the diagonal entries are the same within each matrix. Moreover, the $\lambda_i$'s (or equivalently $L_i$'s) that appear in $\widetilde{A}_0,\widehat{A}_0,\overline{A}_1,\overline{A}_2$ do not intersect with the $\lambda_i$'s that appear in $\overline{A}_3,\overline{A}_4,\dots,\overline{A}_{n-1}$.
The method we used to prove the MDS property of our MSR code construction in Section~\ref{MDS} can be easily generalized to show that \eqref{eq:repair0} also defines an MDS array code $(\widetilde{C}_0,\widehat{C}_0,\overline{C}_1,\overline{C}_2,\dots,\overline{C}_{n-1})$ with code length $n+1$ and code dimension $k+1$. Therefore, $\widetilde{C}_0$ and $\widehat{C}_0$ can be recovered from any $k+1$ vectors in the set $\{\overline{C}_1,\overline{C}_2,\dots,\overline{C}_{n-1}\}$. Once we know the values of $\widetilde{C}_0$ and $\widehat{C}_0$, we are able to recover all the coordinates of $C_0$.

{\bf Second case: How to repair $C_1$.} It is similar to the first case, but here we need to look at the block rows whose indices lie in the set $\{1,3,5,\cdots,2(\ell/2-1)+1\}$. 
These $\ell/2$ block rows of parity check equations can be organized in the matrix form
\begin{equation} \label{eq:repair1}
\overline{A}_0 \overline{C}_0 + \widetilde{A}_1 \widetilde{C}_1 + \widehat{A}_1 \widehat{C}_1
+ \sum_{i=2}^{n-1} \overline{A}_i \overline{C}_i = \mathbf{0} ,
\end{equation}
where $\overline{A}_0, \widetilde{A}_1, \widehat{A}_1, \overline{A}_i, 2\le i\le n-1$ are all $\ell/2\times \ell/2$ matrices, and $\overline{C}_0, \widetilde{C}_1, \widehat{C}_1, \overline{C}_i, 2\le i\le n-1$ are all column vectors of length $\ell/2$. More specifically, the matrices in \eqref{eq:repair1} are for every $0\le a,b\le \ell/2-1$, 
\begin{equation*}
\begin{aligned}
\overline{A}_{0}(a,b) &=\left\{
\begin{array}{ll}
  L_1   & \mbox{if~} a=b, \\
  \mathbf{0} & \mbox{otherwise},
\end{array}
\right.  \\
\widetilde{A}_{1}(a,b) &=\left\{
\begin{array}{ll}
  -L_2   & \mbox{if~} a=b, \\
  \mathbf{0} & \mbox{otherwise},
\end{array}
\right.  \\
\widehat{A}_{1}(a,b) &=\left\{
\begin{array}{ll}
  L_3   & \mbox{if~} a=b, \\
  \mathbf{0} & \mbox{otherwise},
\end{array}
\right.  \\
\overline{A}_{2}(a,b) &=\left\{
\begin{array}{ll}
  L_5   & \mbox{if~} a=b, \\
  \mathbf{0} & \mbox{otherwise},
\end{array}
\right.  
\end{aligned}
\end{equation*}
and for every $1\le i\le n/3-1$,
\begin{equation*}
\begin{aligned}
\overline{A}_{3i}(a,b) &=\left\{
\begin{array}{ll}
  L_{6i+a_{i-1}}   & \mbox{if~} a=b, \\
  L_{6i}-L_{6i+1}   & \mbox{if~} a_{i-1}=0, b_{i-1}=1, \\&\mbox{and~} a_j=b_j \,\forall j\neq i-1, \\
  \mathbf{0} & \mbox{otherwise},
\end{array}
\right.  \\
\overline{A}_{3i+1}(a,b) &=\left\{
\begin{array}{ll}
  L_{6i+2+a_{i-1}}   & \mbox{if~} a=b, \\
  L_{6i+3}-L_{6i+2}   & \mbox{if~} a_{i-1}=1, b_{i-1}=0, \\&\mbox{and~} a_j=b_j \,\forall j\neq i-1, \\
  \mathbf{0} & \mbox{otherwise},
\end{array}
\right. \\
\overline{A}_{3i+2}(a,b) &=\left\{
\begin{array}{ll}
  L_{6i+4+a_{i-1}}   & \mbox{if~} a=b, \\
  \mathbf{0} & \mbox{otherwise},
\end{array}
\right.
\end{aligned}
\end{equation*}
where $\mathbf{0}$ in the last line denotes the all-zero column vector of length $r$, 
and the column vectors in \eqref{eq:repair1} are for every $0\le a\le \ell/2-1$,
\begin{equation*}
\begin{aligned}
\widetilde{C}_{1}(a) &= C_1(2a), \\
\widehat{C}_{1}(a) &= C_1(2a)+C_1(2a+1),
\end{aligned}
\end{equation*}
and for every $0\le i\le n-1$ and $i\neq 1$,
$$\overline{C}_i(a) =C_i(2a+1).$$
Here we still make an important observation: The matrices $\overline{A}_3,\overline{A}_4,\dots,\overline{A}_{n-1}$ are precisely the $n-3$ parity check matrices that would appear in our MSR code construction with code length $n-3$ and subpacketization $\ell/2=2^{n/3-1}=2^{(n-3)/3}$. The other $4$ matrices $\overline{A}_0,\widetilde{A}_1,\widehat{A}_1,\overline{A}_2$ are all block-diagonal matrices, and the diagonal entries are the same within each matrix. Moreover, the $\lambda_i$'s (or equivalently $L_i$'s) that appear in $\overline{A}_0,\widetilde{A}_1,\widehat{A}_1,\overline{A}_2$ do not intersect with the $\lambda_i$'s that appear in $\overline{A}_3,\overline{A}_4,\dots,\overline{A}_{n-1}$.
The method we used to prove the MDS property of our MSR code construction in Section~\ref{MDS} can be easily generalized to show that \eqref{eq:repair1} also defines an MDS array code $(\overline{C}_0,\widetilde{C}_1,\widehat{C}_1,\overline{C}_2,\dots,\overline{C}_{n-1})$ with code length $n+1$ and code dimension $k+1$. Therefore, $\widetilde{C}_1$ and $\widehat{C}_1$ can be recovered from any $k+1$ vectors in the set $\{\overline{C}_0,\overline{C}_2,\dots,\overline{C}_{n-1}\}$. Once we know the values of $\widetilde{C}_1$ and $\widehat{C}_1$, we are able to recover all the coordinates of $C_1$.

{\bf Third case: How to repair $C_2$.}
In order to repair $C_2$, we sum up $\ell/2$ pairs of block rows of each matrix $A_i$ in equation \eqref{eq:repair}. Specifically, we sum up $(2i)$th block row and the $(2i+1)$th block row for every $i\in \{0,1,\cdots,\ell/2-1\}$. In this way, we obtain $\ell/2$ block rows of parity check equations from the original $\ell$ block rows of parity check equations in \eqref{eq:repair}. These $\ell/2$ block rows of parity check equations can be written in the matrix form
\begin{equation} \label{eq:repair2}
\overline{A}_0 \overline{C}_0 + \overline{A}_1 \overline{C}_1 + \widetilde{A}_2 \widetilde{C}_2 + \widehat{A}_2 \widehat{C}_2
+ \sum_{i=3}^{n-1} \overline{A}_i \overline{C}_i = \mathbf{0} ,
\end{equation}
where $\overline{A}_0, \overline{A}_1, \widetilde{A}_2, \widehat{A}_2, \overline{A}_i, 3\le i\le n-1$ are all $\ell/2\times \ell/2$ matrices, and $\overline{C}_0, \overline{C}_1, \widetilde{C}_2, \widehat{C}_2, \overline{C}_i, 3\le i\le n-1$ are all column vectors of length $\ell/2$. More specifically, the matrices in \eqref{eq:repair2} are for every $0\le a,b\le \ell/2-1$, 
\begin{equation*}
\begin{aligned}
\overline{A}_{0}(a,b) &=\left\{
\begin{array}{ll}
  L_0   & \mbox{if~} a=b, \\
  \mathbf{0} & \mbox{otherwise},
\end{array}
\right.  \\
\overline{A}_{1}(a,b) &=\left\{
\begin{array}{ll}
  L_3   & \mbox{if~} a=b, \\
  \mathbf{0} & \mbox{otherwise},
\end{array}
\right.  \\
\widetilde{A}_{2}(a,b) &=\left\{
\begin{array}{ll}
  L_4   & \mbox{if~} a=b, \\
  \mathbf{0} & \mbox{otherwise},
\end{array}
\right.  \\
\widehat{A}_{2}(a,b) &=\left\{
\begin{array}{ll}
  L_5   & \mbox{if~} a=b, \\
  \mathbf{0} & \mbox{otherwise},
\end{array}
\right.  
\end{aligned}
\end{equation*}
and for every $1\le i\le n/3-1$,
\begin{equation*}
\begin{aligned}
\overline{A}_{3i}(a,b) &=\left\{
\begin{array}{ll}
  L_{6i+a_{i-1}}   & \mbox{if~} a=b, \\
  L_{6i}-L_{6i+1}   & \mbox{if~} a_{i-1}=0, b_{i-1}=1, \\&\mbox{and~} a_j=b_j \,\forall j\neq i-1, \\
  \mathbf{0} & \mbox{otherwise},
\end{array}
\right.  \\
\overline{A}_{3i+1}(a,b) &=\left\{
\begin{array}{ll}
  L_{6i+2+a_{i-1}}   & \mbox{if~} a=b, \\
  L_{6i+3}-L_{6i+2}   & \mbox{if~} a_{i-1}=1, b_{i-1}=0, \\&\mbox{and~} a_j=b_j \,\forall j\neq i-1, \\
  \mathbf{0} & \mbox{otherwise},
\end{array}
\right. \\
\overline{A}_{3i+2}(a,b) &=\left\{
\begin{array}{ll}
  L_{6i+4+a_{i-1}}   & \mbox{if~} a=b, \\
  \mathbf{0} & \mbox{otherwise},
\end{array}
\right.
\end{aligned}
\end{equation*}
where $\mathbf{0}$ in the last line denotes the all-zero column vector of length $r$, 
and the column vectors in \eqref{eq:repair2} are for every $0\le a\le \ell/2-1$,
\begin{equation*}
\begin{aligned}
\widetilde{C}_{2}(a) &= C_2(2a),\\
\widehat{C}_{2}(a) &= C_2(2a+1), 
\end{aligned}
\end{equation*}
and for every $0\le i\le n-1$ and $i\neq 2$,
$$\overline{C}_i(a) =C_i(2a)+C_i(2a+1).$$
The matrices $\overline{A}_3,\overline{A}_4,\dots,\overline{A}_{n-1}$ are precisely the $n-3$ parity check matrices that would appear in our MSR code construction with code length $n-3$ and subpacketization $\ell/2=2^{n/3-1}=2^{(n-3)/3}$. The other $4$ matrices $\overline{A}_0,\overline{A}_1,\widetilde{A}_2,\widehat{A}_2$ are all block-diagonal matrices, and the diagonal entries are the same within each matrix. Moreover, the $\lambda_i$'s (or equivalently $L_i$'s) that appear in $\overline{A}_0,\overline{A}_1,\widetilde{A}_2,\widehat{A}_2$ do not intersect with the $\lambda_i$'s that appear in $\overline{A}_3,\overline{A}_4,\dots,\overline{A}_{n-1}$.
The method we used to prove the MDS property of our MSR code construction in Section~\ref{MDS} can be easily generalized to show that \eqref{eq:repair2} also defines an MDS array code $(\overline{C}_0,\overline{C}_1,\widetilde{C}_2,\widehat{C}_2,\overline{C}_3,\dots,\overline{C}_{n-1})$ with code length $n+1$ and code dimension $k+1$. Therefore, $\widetilde{C}_2$ and $\widehat{C}_2$ can be recovered from any $k+1$ vectors in the set $\{\overline{C}_0,\overline{C}_1,\overline{C}_3,\dots,\overline{C}_{n-1}\}$. Once we know the values of $\widetilde{C}_2$ and $\widehat{C}_2$, we are able to recover all the coordinates of $C_2$.

In summary, the matrices $A_i$ in the parity check matrix \eqref{eq:repair} have $\ell$ block rows. Yet the repair of a single failed node only relies on $\ell/2$ block rows of parity check equations. More precisely, for every $0\le i\le n/3-1$, the repair of $C_{3i}$ only involves the block rows whose indices lie in the set $\{a:0\le a\le \ell-1,a_i=0\}$; the repair of $C_{3i+1}$ only involves the block rows whose indices lie in the set $\{a:0\le a\le \ell-1,a_i=1\}$. The repair of $C_{3i+2}$ is slightly more complicated. In this case, we divide the $\ell$ block rows into $\ell/2$ pairs $\{(a,a+2^i):0\le a\le \ell-1,a_i=0\}$ and we sum up the two block rows in each pair. In this way, we obtain $\ell/2$ block rows of parity check equations, and the repair of $C_{3i+2}$ only involves these equations.

%\section{New code construction for k+2} \label{construction2}

%\section{Discussion} 
%Optimal access.

%\bibliographystyle{IEEEtran}
%\bibliography{repair}
% Generated by IEEEtran.bst, version: 1.14 (2015/08/26)

\appendix
\section{} \label{permutation}
Recall that in Section \ref{construction} we write the $n/3$-digit binary expansion of $a$ as $a=(a_{n/3-1},a_{n/3-2},\dots,a_0)$ 
for $\ell=2^{n/3}$ and $a\in\{0,1,\dots,\ell-1\}$.
The $(n,k)$ array code we defined in Section \ref{construction} is
{\footnotesize $$\mathcal{C}=\{(C_0,C_1,\ldots,C_{n-1}): A_0C_0+A_1C_1+\cdots+A_{n-1}C_{n-1}=\mathbf{0}\}.$$}%
The $n$ nodes $(C_0,C_1,\ldots,C_{n-1})$ are divided into $n/3$ groups of size $3$.
Next we show that we can exchange nodes in group $i$ with nodes in group $j$ in the following way without destructing the MDS property.

For the index of two groups $0\leq i<j\leq n/3-1$, we define a function $P_{i\leftrightarrow j}:\{0,1,\dots,\ell-1\}\rightarrow \{0,1,\dots,\ell-1\}$ by
{\small\begin{align*}
&P_{i\leftrightarrow j}(a)\\
=&P_{i\leftrightarrow j}((a_{n/3-1},\cdots,a_{j+1},a_j,a_{j-1},\cdots,a_{i+1},a_i,a_{i-1},\dots,a_0)) \\
=&(a_{n/3-1},\cdots,a_{j+1},a_i,a_{j-1},\cdots,a_{i+1},a_j,a_{i-1},\dots,a_0) \\
=&\overline{a}\\
=&(\overline{a}_{n/3-1},\cdots,\overline{a}_{j+1},\overline{a}_j,\overline{a}_{j-1},\cdots,\overline{a}_{i+1},\overline{a}_i,\overline{a}_{i-1},\dots,\overline{a}_0). 
\end{align*}}%
%That is for $s\in\{0,1,\cdots,n/3-1\}$ and $s\neq i,j$, $\overline{a}_s=a_s$, $\overline{a}_i=a_j,\overline{a}_j=a_i$. \\
%By the parity check equation \eqref{eq:pc}, we have 
%$$A_0C_0+A_1C_1+\cdots+A_{n-1}C_{n-1}=0. $$
%Write in groups as follows: 
%\begin{align}\label{eq:permu1}
%&A_0C_0+A_1C_1+A_2C_2+
%\cdots \nonumber \\
%+&A_{3i}C_{3i}+A_{3i+1}C_{3i+1}+A_{3i+2}C_{3i+2}+\cdots+A_{3j}C_{3j}+A_{3j+1}C_{3j+1}+A_{3j+2}C_{3j+2}+\cdots \nonumber \\
%+&A_{3\times(n/3-1)}C_{3\times(n/3-1)}+
%A_{3\times(n/3-1)+1}C_{3\times(n/3-1)+1}+A_{3\times(n/3-1)+2}C_{3\times(n/3-1)+2}=0.
%\end{align}
Now let us define another $(n,k)$ array code by
{\small $$\overline{\mathcal{C}}=\{(C_0,C_1,\ldots,C_{n-1}): \overline{A}_0C_0+\overline{A}_1C_1+\cdots+\overline{A}_{n-1}C_{n-1}=\mathbf{0}\},$$}%
where
\begin{align*}
    \overline{A}_k(a,b)&=A_k(P_{i\leftrightarrow j}(a),P_{i\leftrightarrow j}(b))=A_k(\overline{a},\overline{b}),
\end{align*}
 for $a,b\in\{0,1,\dots,\ell-1\}$ and $k\in\{0,1,\cdots,n-1\}$.
We can readily check that if $(C_0,C_1,\ldots,C_{n-1})\in\mathcal{C}$,
then $(\overline{C}_0,\overline{C}_1,\ldots,\overline{C}_{n-1})\in\overline{\mathcal{C}}$ where
$\overline{C}_k(a)=C_k(P_{i\leftrightarrow j}(a))=C_k(\overline{a})$ for $0\leq k\leq n-1$ and $0\leq a\leq \ell-1$,
and vice versa. Hence $\mathcal{C}$ and $\overline{\mathcal{C}}$ are permutation equivalent.
\if{false}
Now we use the function $P_{i\leftrightarrow j}$ to convert it into the following form 
$$\overline{A}_0\overline{C}_0+\overline{A}_1\overline{C}_1+\cdots+\overline{A}_{n-1}\overline{C}_{n-1}=0,$$
where 
\begin{align*}
    \overline{A}_k(a,b)&=A_k(P_{i\leftrightarrow j}(a),P_{i\leftrightarrow j}(b))=A_k(\overline{a},\overline{b}),\\
    \overline{C}_k(a)&=C_k(P_{i\leftrightarrow j}(a))=C_k(\overline{a}).
\end{align*}
Write in groups as follows: 
\begin{align}\label{eq:permu2}
&\overline{A}_0\overline{C}_0+\overline{A}_1\overline{C}_1+\overline{A}_2\overline{C}_2+
\cdots \nonumber \\
+&\overline{A}_{3i}\overline{C}_{3i}+\overline{A}_{3i+1}\overline{C}_{3i+1}+\overline{A}_{3i+2}\overline{C}_{3i+2}+\cdots+\overline{A}_{3j}\overline{C}_{3j}+\overline{A}_{3j+1}\overline{C}_{3j+1}+\overline{A}_{3j+2}\overline{C}_{3j+2}+\cdots \nonumber \\
+&\overline{A}_{3\times(n/3-1)}\overline{C}_{3\times(n/3-1)}+
\overline{A}_{3\times(n/3-1)+1}\overline{C}_{3\times(n/3-1)+1}+\overline{A}_{3\times(n/3-1)+2}\overline{C}_{3\times(n/3-1)+2}=0.
\end{align}
\fi

By definition, for the group $s\in\{0,1,\cdots,n/3-1\}$ and $s\neq i,j$,  we have $\overline{a}_s=a_s$ and
{\footnotesize \begin{align*}
\overline{A}_{3s}(a,b)=A_{3s}(\overline{a},\overline{b}) &=\left\{
\begin{array}{ll}
  L_{6s+\overline{a}_s}   & \mbox{if~} \overline{a}=\overline{b},  \\
  L_{6s}-L_{6s+1}   & \mbox{if~} \overline{a}_s=0, \overline{b}_s=1,  \\&\mbox{and~} \overline{a}_t=\overline{b}_t,\, \forall t\neq s,  \\
  \mathbf{0} & \mbox{otherwise},
\end{array} 
\right. \\
&=\left\{
\begin{array}{ll}
  L_{6s+a_s}   & \mbox{if~} a=b,  \\
  L_{6s}-L_{6s+1}   & \mbox{if~} a_s=0, b_s=1, \\&\mbox{and~} a_t=b_t,\, \forall t\neq s,  \\
  \mathbf{0} & \mbox{otherwise},
\end{array}
\right.  \\
&=A_{3s}(a,b),\\
\overline{A}_{3s+1}(a,b)=A_{3s+1}(\overline{a},\overline{b}) 
%&=\left\{
%\begin{array}{ll}
%  L_{6s+2+\overline{a}_s}   & \mbox{if~} \overline{a}=\overline{b} ; \\
%  L_{6s+3}-L_{6s+2}   & \mbox{if~} \overline{a}_s=1, \overline{b}_s=0, \\
%  & \mbox{and~} \overline{a}_t=\overline{b}_t, \forall t\neq s ; \\
%  0 & \mbox{otherwise}.
%\end{array} 
%\right. \\
%&=\left\{
%\begin{array}{ll}
%  L_{6s+2+a_s}   & \mbox{if~} a=b ; \\
%  L_{6s+3}-L_{6s+2}   & \mbox{if~} a_s=1, b_s=0, \\
%  & \mbox{and~} a_t=b_t, \forall t\neq s ; \\
%  0 & \mbox{otherwise}.
%\end{array} 
%\right. \\
&=A_{3s+1}(a,b),\\
\overline{A}_{3s+2}(a,b)=A_{3s+2}(\overline{a},\overline{b}) 
%&=\left\{
%\begin{array}{ll}
%  L_{6s+4+\overline{a}_s}   & \mbox{if~} \overline{a}=\overline{b} ; \\
%  0 & \mbox{otherwise}.
%\end{array} 
%\right. 
%=\left\{
%\begin{array}{ll}
%  L_{6s+4+a_s}   & \mbox{if~} a=b ; \\
%  0 & \mbox{otherwise}.
%\end{array} 
%\right. \\
&=A_{3s+2}(a,b).
\end{align*}}%

Note that $\overline{a}_i=a_j$, so for $s=i$, we have
{\footnotesize \begin{align*}
\overline{A}_{3i}(a,b)=A_{3i}(\overline{a},\overline{b}) &=\left\{
\begin{array}{ll}
  L_{6i+\overline{a}_i}   & \mbox{if~} \overline{a}=\overline{b}, \\
  L_{6i}-L_{6i+1}   & \mbox{if~} \overline{a}_i=0, \overline{b}_i=1, \\&\mbox{and~} \overline{a}_t=\overline{b}_t\, \forall t\neq i, \\
  \mathbf{0} & \mbox{otherwise},
\end{array}
\right. \\
&=\left\{
\begin{array}{ll}
  L_{6i+a_j}   & \mbox{if~} a=b, \\
  L_{6i}-L_{6i+1}   & \mbox{if~} a_j=0, b_j=1, \\&\mbox{and~} a_t=b_t\, \forall t\neq j, \\
  \mathbf{0} & \mbox{otherwise},
\end{array} %\text{this is the same structure as $A_{3j}$.}
\right. \\
\overline{A}_{3i+1}(a,b)=A_{3i+1}(\overline{a},\overline{b}) 
%&=\left\{
%\begin{array}{ll}
%  L_{6i+2+\overline{a}_i}   & \mbox{if~} \overline{a}=\overline{b} ; \\
%  L_{6i+3}-L_{6i+2}   & \mbox{if~} \overline{a}_i=1, \overline{b}_i=0, \\
%  & \mbox{and~} \overline{a}_t=\overline{b}_t \forall t\neq i ; \\
%  0 & \mbox{otherwise}.
%\end{array} 
%\right. \\
&=\left\{
\begin{array}{ll}
  L_{6i+2+a_j}   & \mbox{if~} a=b \\
  L_{6i+3}-L_{6i+2}   & \mbox{if~} a_j=1, b_j=0, \\&\mbox{and~} a_t=b_t \,\forall t\neq j \\
  \mathbf{0} & \mbox{otherwise},
\end{array} %\text{this is the same structure as $A_{3j+1}$.}
\right. \\
\overline{A}_{3i+2}(a,b)=A_{3i+2}(\overline{a},\overline{b}) 
%&=\left\{
%\begin{array}{ll}
%  L_{6i+4+\overline{a}_i}   & \mbox{if~} \overline{a}=\overline{b} ; \\
%  0 & \mbox{otherwise}.
%\end{array}
%\right. \\
&=\left\{
\begin{array}{ll}
  L_{6i+4+a_j}   & \mbox{if~} a=b, \\
  \mathbf{0} & \mbox{otherwise}.
\end{array} %\text{this is the same structure as $A_{3j+2}$.}
\right.
\end{align*}}%

Similarly we have $\overline{a}_j=a_i$ and
{\footnotesize \begin{align*}
\overline{A}_{3j}(a,b)=A_{3j}(\overline{a},\overline{b}) 
%&=\left\{
%\begin{array}{ll}
%  L_{6j+\overline{a}_j}   & \mbox{if~} \overline{a}=\overline{b}  \\
%  L_{6j}-L_{6j+1}   & \mbox{if~} \overline{a}_j=0, \overline{b}_j=1, \mbox{and~} \overline{a}_t=\overline{b}_t\, \forall t\neq j  \\
%  0 & \mbox{otherwise}
%\end{array}
%\right. \\
&=\left\{
\begin{array}{ll}
  L_{6j+a_i}   & \mbox{if~} a=b,  \\
  L_{6j}-L_{6j+1}   & \mbox{if~} a_i=0, b_i=1, \\&\mbox{and~} a_t=b_t\, \forall t\neq i,  \\
  \mathbf{0} & \mbox{otherwise},
\end{array} %\text{this is the same structure as $A_{3i}$.}
\right. \\
\overline{A}_{3j+1}(a,b)=A_{3j+1}(\overline{a},\overline{b}) 
%&=\left\{
%\begin{array}{ll}
%  L_{6j+2+\overline{a}_j}   & \mbox{if~} \overline{a}=\overline{b} ; \\
%  L_{6j+3}-L_{6j+2} \\  & \mbox{if~} \overline{a}_j=1, \overline{b}_j=0, \\
%  & \mbox{and~} \overline{a}_t=\overline{b}_t \forall t\neq j ; \\
%  0 & \mbox{otherwise}.
%\end{array} 
%\right. \\
&=\left\{
\begin{array}{ll}
  L_{6j+2+a_i}   & \mbox{if~} a=b \\
  L_{6j+3}-L_{6j+2}   & \mbox{if~} a_i=1, b_i=0, \\&\mbox{and~} a_t=b_t\, \forall t\neq i \\
  \mathbf{0} & \mbox{otherwise},
\end{array} %\text{this is the same structure as $A_{3i+1}$.}
\right. \\
\overline{A}_{3j+2}(a,b)=A_{3j+2}(\overline{a},\overline{b}) 
%&=\left\{
%\begin{array}{ll}
%  L_{6j+4+\overline{a}_j}   & \mbox{if~} \overline{a}=\overline{b} ; \\
%  0 & \mbox{otherwise}.
%\end{array}
%\right. \\
&=\left\{
\begin{array}{ll}
  L_{6j+4+a_i}   & \mbox{if~} a=b, \\
  \mathbf{0} & \mbox{otherwise}.
\end{array} %\text{this is the same structure as $A_{3i+2}$.}
\right.
\end{align*}}%

Now we define 
{\small $$\widehat{\mathcal{C}}=\{(C_0,C_1,\ldots,C_{n-1}): \widehat{A}_0C_0+\widehat{A}_1C_1+\cdots+\widehat{A}_{n-1}C_{n-1}=\mathbf{0}\},$$}%
where 
$$\widehat{A}_{3s}=\overline{A}_{3s},\ \widehat{A}_{3s+1}=\overline{A}_{3s+1},\ \widehat{A}_{3s+2}=\overline{A}_{3s+2}, \text{ for }s\neq i,j,$$
and 
$$\widehat{A}_{3i}=\overline{A}_{3j},\ \widehat{A}_{3i+1}=\overline{A}_{3j+1},\ \widehat{A}_{3i+2}=\overline{A}_{3j+2},$$
$$\widehat{A}_{3j}=\overline{A}_{3i},\ \widehat{A}_{3j+1}=\overline{A}_{3i+1},\ \widehat{A}_{3j+2}=\overline{A}_{3i+2}.$$

We can easily check that
{\tiny \begin{align*}
    \widehat{\mathcal{C}}=\{&(C_0,C_1,\ldots,C_{3j},C_{3j+1},C_{3j+2},\ldots,C_{3i},C_{3i+1},C_{3i+2},\ldots,C_{n-1}):\\
    &(C_0,C_1,\ldots,C_{3i},C_{3i+1},C_{3i+2},\ldots,C_{3j},C_{3j+1},C_{3j+2},\ldots,C_{n-1})\in\overline{\mathcal{C}}\},
\end{align*}}%
which says that $\widehat{\mathcal{C}}$ is permutation equivalent to $\overline{\mathcal{C}}$.

From the above, we can see that $\widehat{A}_{3j}(=\overline{A}_{3i})$, $\widehat{A}_{3j+1}(=\overline{A}_{3i+1})$ and $\widehat{A}_{3j+2}(=\overline{A}_{3i+2})$ are defined in 
the same way as $A_{3j}$, $A_{3j+1}$ and $A_{3j+2}$, but using $\lambda_{6i},\cdots,\lambda_{6i+5}$ instead of $\lambda_{6j},\cdots,\lambda_{6j+5}$.
Similarly $\widehat{A}_{3i}(=\overline{A}_{3j})$, $\widehat{A}_{3i+1}(=\overline{A}_{3j+1})$ and $\widehat{A}_{3i+2}(=\overline{A}_{3j+2})$ are defined in 
the same way  as $A_{3i}$, $A_{3i+1}$ and $A_{3i+2}$, 
but using $\lambda_{6j},\cdots,\lambda_{6j+5}$ instead of $\lambda_{6i},\cdots,\lambda_{6i+5}$.
Therefore any repair scheme works for $\mathcal{C}$
will also work for $\widehat{\mathcal{C}}$, and vice versa. In summary, we can exchange group $i$ with group $j$ as above without destructing the MDS property.
At last we point out any erasure pattern can be converted to the canonical erasure pattern using the above operation recursively. 

\end{document}